 \documentclass[11pt, reqno]{amsart}

\usepackage[utf8x]{inputenc}

\usepackage{amsmath, amsthm, amsfonts, amssymb,mathtools}

\usepackage{epsfig, enumerate}
\usepackage{color}

\usepackage{pdf14}

\newtheorem{Theo}{Theorem}

\newtheorem{proposition}{Proposition}
\newtheorem{lemma}[proposition]{Lemma}

\newtheorem*{rem*}{Remarks}
\newtheorem{rem}[proposition]{Remark}

\newcommand{\C}{{\mathbf C}}

\newcommand{\vv}{{\vec v}}

\newcommand{\Aa}{{\mathcal A}}

\newcommand{\Cc}{{\mathcal C}}

\newcommand{\Pp}{{\mathcal P}}
\newcommand{\Qq}{{\mathcal Q}}

\newcommand{\Ss}{{\mathcal S}}

\newcommand{\Uu}{{\mathcal U}}
\newcommand{\Vv}{{\mathcal V}}
\newcommand{\Ww}{{\mathcal W}}

\newcommand{\Zz}{{\mathcal Z}}

\newcommand{\CC}{{\mathbb C}}
\newcommand{\DD}{{\mathbb D}}
\newcommand{\EE}{{\mathbb E}}

\newcommand{\GG}{{\mathbb G}}
\newcommand{\HH}{{\mathbb H}}

\newcommand{\NN}{{\mathbb N}}

\newcommand{\PP}{{\mathbb P}}

\newcommand{\RR}{{\mathbb R}}
\newcommand{\Sb}{{\mathbb S}}
\newcommand{\TT}{{\mathbb T}}
\newcommand{\UU}{{\mathbb U}}

\newcommand{\ZZ}{{\mathbb Z}}

\newcommand{\aaa}{{\mathfrak a}}
\newcommand{\bbb}{{\mathfrak b}}
\newcommand{\ccc}{{\mathfrak c}}
\newcommand{\ddd}{{\mathfrak d}}

\newcommand{\one}{\mathbf{I}}
\newcommand{\nul}{\mathbf{0}}

\newcommand{\qtx}[1]{\quad\text{#1}\quad}

\newcommand{\GL}{{\rm GL}}

\newcommand{\pmat}[1]{\begin{pmatrix} #1  \end{pmatrix}}
\newcommand{\smat}[1]{\left( \begin{smallmatrix} #1  \end{smallmatrix} \right)}

\DeclareMathOperator{\Tr}{{\rm Tr}}

\DeclareMathOperator{\supp}{{\rm supp}}

\numberwithin{proposition}{section}
\numberwithin{equation}{section}

\usepackage{hyperref}

\title[On a.c. spectrum for 1-channel unitary operators]{On absolutely continuous spectrum for one-channel unitary operators}

\author{Olivier Bourget, Gregorio Moreno, Christian Sadel, Amal Taarabt}

\address{Facultad de Matem\'atcas, Pontificia Universidad Cat\'olica de Chile} 

\email[O. Bourget]{bourget@uc.cl}
\email[G. Moreno]{grmoreno@uc.cl}
\email[C. Sadel]{chsadel@uc.cl}  
\email[A. Taarabt]{amtaarabt@uc.cl}

\thanks{corresponding author: C. Sadel.}


\begin{document}

\begin{abstract}
In this paper, we develop the radial transfer matrix formalism for unitary one-channel operators. This generalizes previous formalisms for CMV matrices and scattering zippers. 
We establish an analog of Carmona's formula and deduce criteria for absolutely continuous spectrum which we apply to random Hilbert Schmidt perturbations of periodic scattering zippers. 
\end{abstract}

\maketitle



\section{Setup and result}

In recent years there has been some interest in the use of  `reduced transfer matrices' for Schrödinger operators, where the hopping between `shells' or  `slices' of some graph are of fixed (but not full) rank \cite{DC,Sa-AT,Sa-OC,SB}. 
These matrices were introduced independently by Dwivedi-Chua \cite{DC} and Sadel \cite{Sa-AT, Sa-OC,Sa-Tr}.
In the strip case, `reduced' means for instance that the dimension of the transfer matrices is smaller than the usual `twice of the strip width', 
it is only `twice the rank'. In the most extreme case, where the connections are only of rank 1, we have $2 \times 2$ transfer matrices.
Such Hermitian operators were called `one-channel operators' in \cite{Sa-OC}.
This can occur even if the size of the  `shells' grow, leading to graphs with a radial growth corresponding to some higher dimensions.
The size reduction of the transfer matrices can help to analyze certain spectral aspects more easily, see
for instance  \cite{Sa-AT,Sa-OC}.

In the Hermitian case, the general formalism of transfer matrices started with one-channel operators and similar models and has been extended by Sadel to any locally finite hopping operator, even if the rank of the connections grow\footnote{including usual Schrödinger operators on $\ZZ^d$ for any dimension $d$} \cite{Sa-Tr}. 
The core of \cite{Sa-OC, Sa-Tr} is the generalization of a spectral averaging formula using transfer matrices originally found by Carmona  for discrete Jacobi and one-dimensional continuous Schrödinger operators \cite{Car, Cala}.  In the most general case of Hermitian locally finite hopping operators, one works with  `sets of rectangular transfer matrices'.
After generalizing Carmona's formula, following arguments by Last-Simon \cite{Lasi} one can obtain criteria for delocalization (absolutely continuous spectrum) which had been applied to various random models 
 \cite{BMT, GS, KiLS, Sa-AT, Sa-OC, Sa-Tr}.

There is a natural correspondence between the formalisms of CMV and Jacobi matrices. 
Both cases use transfer matrices and orthogonal polynomials to analyze their spectral theory \cite{Si-OPUC, Si-Szego}. This suggests that there might be some analogue of \cite{Sa-Tr} in the unitary set-up and as a first step we are investigating the analogue of the one-channel case.  In fact, the analogues of Carmona's formula and Last-Simon's criterion for CMV matrices are mentioned in \cite{Si-OPUC} where the results are based on orthogonal polynomials on the unit circle.
Here, we get more general versions of these theorems through an operator theoretic point of view.

Obtaining criteria for absolutely continuous spectrum and  applications to disordered system is another motivation for this work. The most important model for a disordered quantum system is the so called Anderson model in the Hermitian case, introduced in \cite{Anderson}, and there are unitary analogues \cite{HJS1,HJS2,HS,Kos}.
The so called Anderson localization is now well understood, see e.g. \cite{AM, AW-book, BK, B-etal, GK2, JX, KlLS, Li, LZ} and references therein for the Hermitian case, and \cite{ABJ1,ABJ2,BHJ,B-etal, JM,Si-CMV, Zhu} for the unitary case. However, proving the conjectured delocalization for small disorder in 3 and higher dimensional models remains a major open problem in the field and one needs to develop new techniques.
Since the criteria for delocalization coming from Carmona's formula and its generalizations have been used for various interesting models in the Hermitian case \cite{BMT, GS, KiLS, Sa-AT, Sa-OC, Sa-Tr}, it is worthwhile to further generalize it in the unitary world.
We aim to obtain a full unitary analogue of the radial transfer matrix set formalism in \cite{Sa-Tr} for finite hopping unitary operators, including higher dimensional quantum walks in $\ZZ^d$ and Chalker-Coddington models.
Note that other techniques have been used to get delocalization for random operators on tree graphs 
in the Hermitian case \cite{ASW, AW, FHS2, KLW, Kl, KS, Sa-FC, Sa-Fib} and the unitary case \cite{HJ}, and for unitary network models
\cite{ABJ1, ABJ2, ABJ3, ABJ4}.

\vspace{.2cm}

Our approach is based on the formalism for scattering zippers \cite{MSb} which is the most general
 unitary analogue to Jacobi and block-Jacobi operators. It includes CMV matrices, block CMV matrices and (quasi-) one dimensional quantum walks on strips. 
 Even though the models considered in \cite{MSb} correspond to block Jacobi operators with full rank connection, 
the construction of the transfer matrices contains formulas which resemble the transfer matrices in \cite{DC,Sa-OC, Sa-Tr}. 
This makes it an excellent starting point. 

We will in general consider certain unitary operators $\Uu=\Ww\Vv$ on $\ell^2(\GG)$ where $\GG$ is considered as some countable set (or graph).
$\Ww$ and $\Vv$ are unitary operators which are direct sums of finite dimensional unitary matrices, but on different partitions of $\GG$. Using the partitions for the direct sum of $\Vv$, we may consider $\Vv$ as some coin matrix (analogue of `potential') and $\Ww$ as the operator giving a `walk' among the sets for the partition. This way, we may interpret $\Uu$ as a quantum walk. The `one-channel' structure will be implemented by a very particular structure on $\Ww$.
Our formalism includes  one-channel scattering zippers (such as one-dimensional quantum walks and CMV matrices), certain quantum walks on carbon chains and certain one-channel stroboscopic models in higher dimension.

We obtain an analogue of Carmona's formula (cf. Theorem~\ref{th-main}) and of the Last-Simon criterion for absolutely continuous spectrum (cf. Theorem~\ref{th-main2}). This criterion can then be applied to unitary one-channel models with a random decaying $\ell^2$ perturbation (cf. Theorem~\ref{th-main3}). Similar results were proved for Jacobi operators \cite{KiLS}, block-Jacobi operators
\cite{FHS3, GS} and discrete Dirac operators \cite{BMT}.

\vspace{.2cm}

Let us give an overview of the paper.
First, in Section~\ref{sub-oc} we introduce and define the one-channel unitary operators. Then, in Section~\ref{sub-tr} we introduce the transfer matrices and  Section~\ref{sub-sp-av}  states the analogue of Carmona's formula and of Last-Simon's criterion for absolutely continuous spectrum  is (Theo\-rems~\ref{th-main} and \ref{th-main2}). In Section~\ref{sub-random} we state the result on a.c. spectrum for decaying random perturbations of a periodic one-channel scattering zipper (Theorem~\ref{th-main3}).\\
In Section~\ref{sec-examples} we give several examples of unitary one-channel operators. First, we show how ordinary one-dimensional quantum walks can be brought into this framework, then we define generalized one-channel quantum walks like quantum walks on carbon chains, and finally, we give an example of some one-channel stroboscopic dynamics on $\ell^2(\ZZ^2)$. 
Let us note that the latter two examples are not covered by the CMV or scattering zipper formalism.
Section~\ref{sec-tr} establishes the connection between transfer matrix and resolvent which is used to prove Theorem~\ref{th-main} and \ref{th-main2} in Section~\ref{sec-sp-av}. Finally, in Section~\ref{sec-perper} we prove Theorem~\ref{th-main3}.

\subsection{One-channel unitary operators\label{sub-oc}}

First, we consider a partition of $\GG$ into countably many finite sets $\Sb_n$ which we will call `shells', which have at least two points,
$$
\GG\,=\,\bigsqcup_{n=0}^\infty \Sb_n,\qquad 2\leq |\Sb_n| < \infty,
$$
 then we can write
$$
\ell^2(\GG)\,=\,\bigoplus_{n=0}^\infty \ell^2(\Sb_n)=\bigoplus_{n=0}^\infty \CC^{\Sb_n},
$$
(as an orthogonal Hilbert-space sum) and similarly
$$
\Psi=\bigoplus_{n=0}^\infty \Psi_n\,\in\,\ell^2(\GG) \qtx{where} \Psi_n \in \ell^2(\Sb_n)=\CC^{\Sb_n}.
$$
Let us also introduce the finite sub-graphs from level $0$ to $N$ as
$$
\GG_{N} = \bigsqcup_{n=0}^N \Sb_n,
$$
and adopt similar notations as $\psi=\bigoplus_{n=0}^N \psi_n \in \ell^2(\GG_N)=\CC^{\GG_N}$
with $\psi_n \in \CC^{\Sb_n}$.

Mostly,  the direct sum has to be understood as an Hilbert-space orthogonal sum.
However, the operators we consider are of finite hopping type and extend naturally to the set of all functions from $\GG$ to $\CC$,  and we may also use the notation above
for $\Psi \in \CC^\GG$ with $\Psi_n=\Psi|_{\Sb_n}$ being the restriction of $\Psi$ to $\Sb_n$.
In physics literature the space $\CC^\GG$ maybe referred as `generalized states' and solutions to $U\Psi=z\Psi$ for $\Psi\in \CC^\GG$ as `generalized eigenfunctions'.
Furthermore,  for $n\leq N$, an element $\varphi\in\CC^{\Sb_n}$ can also be considered as an element of $\CC^{\GG_{N}}$ or $\ell^2(\GG)$, identifying $\varphi$ with $\varphi\oplus \bigoplus_{k\neq n} \nul$.
Using some adequate basis we will identify $\CC^{\Sb_n}$ with $\CC^{|\Sb_n|}$ later on,  but we prefer the notation $\CC^{\Sb_n}$ to distinguish the spaces for different shells with possibly same number of elements.
Similarly, by notations like $\CC^{\Sb_n \times l}$ we understand the set of linear maps from $\CC^l$ to $\CC^{S_n}$, which (given a basis of $\CC^{\Sb_n}$) can be identified with 
the set of $|\Sb_n| \times l$ matrices, or with the set of $l$-tuples of vectors in $\CC^{\Sb_n}$.
First, we define the operator $\Vv$ by
\begin{equation}\label{eq-V}
\Vv\,=\,\bigoplus_{n=0}^\infty V_n \qtx{where} V_n\,=\Pp_{\Sb_n} \Vv \Pp_{\Sb_n}\, \in\,\UU(\Sb_n)\;,
\end{equation}
where $\Pp_{\Sb_n}$ is the orthogonal projection of $\ell^2(\GG)$ onto $\CC^{\Sb_n}=\ell^2(\Sb_n)$.
Here, we use the standard physics scalar product $\langle \varphi, \varphi'\rangle = \sum_x \bar \varphi(x) \varphi'(x)$, and, moreover,
$\UU(\Sb_n)$ denotes the unitary matrices on $\CC^{\Sb_n}$. Thus, $\Vv$ is unitary.

In order to connect the shells through one channel, we assign a `forward' and `backward' mode $e_{(n,+)},\,e_{(n,-)} \in \CC^{\Sb_n}$ which are orthonormal vectors, that means
$$
\langle e_{(n,\star)}\,,\,e_{(n,\diamond)} \rangle \,=\,e_{(n,\star)}^* e_{(n,\diamond)}\,=\,\delta_{\star, \diamond} \qtx{where} \star,\diamond\, \in\{+,-\}\;.
$$
Furthermore let 
$$
Q_n= (e_{(n,-)},e_{(n,+)}) \,\in\,\CC^{\Sb_n \times 2},\quad P_n=\one_{\Sb_n}-Q_n Q_n^*,
$$
$Q_n Q_n^*$ is the orthogonal projection onto ${\rm span}(e_{(n,-)},e_{(n,+)})$, 
$P_n$ the orthogonal projection on the orthogonal complement within $\CC^{\Sb_n}$ and $\one_{\Sb_n}$ is the identity operator on $\CC^{\Sb_n}$.
Then we define the operators $\Ww^{(u)}$ by
\begin{equation} \label{eq-W}
\Ww^{(u)}\,=\,  u \,  e_{(n,-)} e_{(n,-)}^* + P_{0} +\sum_{n=1}^\infty \left( (e_{(n-1,+)},  \; e_{(n,-)}) W_n \pmat{e_{(n-1,+)}^* \\ e_{(n,-)}^* } + P_n \right).
\end{equation}

Here, $ u \in \UU(1)=\partial \DD$ is some sort of `left boundary condition',  and $W_n \in \UU(2)$, where $\UU(k)$ denotes the unitary operators on $\CC^k$.
In the notation for $\Ww$ above we interpret the vectors $e_{(n,\pm)}$ in $\CC^{\Sb_n}$ as column vectors in $\ell^2(\GG)$ and in the sense of matrices as maps from $\CC$ to $\CC^{\Sb_n}\subset \ell^2(\GG)$, and $e_{(n,\pm)}^*$ as maps from $\CC^{\Sb_n}$ or $\ell^2(\GG)$ to $\CC$.
Also, $P_n$ is naturally interpreted as an operator on $\ell^2(\GG)$ identifying it with $P_n\oplus \bigoplus_{k\neq n} \nul$.
Another way of representing $\Ww$ is by using the projections
$$
\Qq_n=|e_{(n-1,+)}\rangle\, \langle e_{(n-1,+)}|\,+\,|e_{(n,-)}\rangle \,\langle e_{(n,-)}|\qtx{for} n\in\ZZ_+^*, \quad  
\Qq_0=|e_{(0,-)}\rangle\, \langle e_{(0,-)}|
$$
$$\Qq=\sum_{n\in\ZZ_+} \Qq_n $$
then
$$
\Ww^{(u)}\,=\,\Qq^\perp\,\oplus\, u\,\oplus\, \bigoplus_{n\in\ZZ_+^*} W_n, $$
\qtx{where} 
$$
W_n\,=\, \Qq_n \Ww^{(u)} \Qq_n\,\in\,\UU(2) \qtx{or $n\geq 1$ and}
u=\Qq_0 \Ww^{(u)} \Qq_0 \,\in\,\UU(1)\,.
$$
One may consider $u=1$ as the `natural boundary condition'. 
\newline
Using an orthonormal basis of $\CC^{\Sb_n}=\ell^2(\Sb_n)$ where $e_{(n,-)}$ is the first, and $e_{(n,+)}$ the last vector, one may write
$$
\Psi_n=\pmat{\Psi_{(n,-)} \\ \Psi_{(n,0)} \\ \Psi_{(n,+)}} \in \CC^{|\Sb_n|}\qtx{where} \Psi_{(n,\pm)}=e_{(n,\pm)}^* \Psi_n\in\CC\,,\quad 
\Psi_{(n,0)}\in \CC^{|\Sb_n|-2}\;.
$$
Then, using these bases to form an orthonormal basis of $\ell^2(\GG)$ we can represent $\Vv$ and $\Ww^{(u)}$ as semi-infinite diagonal block matrices of the following form
$$
\Vv=\pmat{V_0 \\ & V_1 \\ & & V_2 \\ & & & \ddots}\,; \quad \Ww^{(u)}=\pmat{u & \\ & \one_{|\Sb_0|-2} \\ & & W_1 \\ & & & \one_{|\Sb_1|-2} \\ & & & & W_2 \\
& & & & & \one_{|\Sb_2|-2}\\ & & & & & & \ddots}
$$
where $\one_m$ is the identity matrix of size $m \times m$.
The blocks of the matrix $\Vv$ have sizes $|\Sb_n|$, with $n=0,1,2,\ldots$, where the blocks of $\Ww$ are more refined and of sizes $1, |\Sb_0|-2,2,|\Sb_1|-2, 2,|\Sb_2|-2,2 \ldots$ and so on. The $2 \times 2$ matrix $W_n$ connects the blocks $\Sb_{n-1}$ to $\Sb_n$.
In this representation it is easy to realize that $\Vv$ and $\Ww^{(u)}$ are unitary operators.

Now let
\begin{equation} \label{eq-U}
 \Uu^{(u)}\,=\,\Ww^{(u)} \Vv  \qtx{and} 
\tilde \Uu^{(u)}\,=\,\Vv\,\Ww^{(u)}\;,
\end{equation}
then we call  $(\Uu^{(u)},\,\tilde \Uu^{(u)})$ 
a {\it conjugated pair of one-channel unitary operators}.
Note that $$\tilde \Uu^{(u)} = (\Ww^{(u)})^* \,\Uu\,\Ww\,=\,(\Ww^{(u)})^{-1}\,\Uu\,\Ww^{(u)}\;.$$
If we omit the upper index $(u)$ we consider $u=1$.

Basically, the sequence $W_n$ connecting the vectors $e_{(n,+)},e_{(n+1,-)}$ defines a `channel' through which waves can travel towards infinity (across the shells) when considering the dynamics $\Uu^k, \tilde \Uu^k$, $k\in \NN$, $k \to \infty$.

\vspace{.2cm}

Furthermore, we introduce the restrictions to the finite graph $\GG_N$ with boundary conditions $u,v \in \UU(1)$ by
\begin{equation}\label{eq-U-G_N}
\Uu^{(u,v)}_N=\Ww^{(u,v)}_N\, \Vv_N\;, \quad
\tilde U^{(u,v)}_N\,=\,\Vv_N \Ww^{(u,v)}_N,
\end{equation}
where
\begin{equation*}
 \Ww^{(u,v)}_N\,=\,
\pmat{u & \\ & \one_{|\Sb_0|-2} \\ & & W_1 \\ & & & \ddots  \\ & & & & W_N \\
& & & & & \one_{|\Sb_N|-2} \\ & & & & & & v} \, \in \, \UU(\GG_N)\;
\end{equation*}
and
\begin{equation*}
\Vv_N\,=\,\pmat{V_0 \\ & V_1 \\ & & \ddots \\ & & & V_N}\, \in \, \UU(\GG_N)\;.
\end{equation*}
One may write a one-dimensional random quantum walk in this framework as the operator $\Uu$ using shells of size $|\Sb_n|=2$, (so $P_n=\nul$) 
where $W_n$ are transpositions, $W_n=\pmat{0&1\\1&0}$. The definitions of the transfer matrices are guided by considering this case.
In fact, one may think of $\Uu$ as a one-channel quantum walk and  $V_n$ as `coin' matrices in $\Sb_n$. Then the `walk' happens according to the matrices $W_n$ across the channel $(e_{(n,+)},e_{(n+1,-)})_{n \in \NN_0}$.
Furthermore, in the case $|\Sb_n|=2$ the operator $\tilde \Uu$ corresponds exactly to the scattering zippers with size $L=2$ in \cite{MSb}, which include
 CMV matrices and one-dimensional quantum walks. Thus, our framework can be seen as an extension of both.

\begin{rem}
\mbox{}
\begin{enumerate}[{\rm (i)}]
\item Many calculations in this article in fact work for $l$-channel unitary operators.  In this case, $e_{(n,+)}, e_{(n,-)} \in \CC^{S_n \times l}$ are orthogonal partial isometries,
meaning $e_{(n,\star)}^* e_{(n,\diamond)}=\delta_{\star,\diamond} \,\one_l$, implying that ${\rm Ran}\, e_{(n,+)}^*$ and ${\rm Ran}\, e_{(n,-)}^*$ are orthogonal $l$-dimensional subspaces of $\CC^{S_n}$.  Then, the boundary conditions $u,v$, as well as the $W_n$ will  be $2l \times 2l$ unitary matrices. These type of operators do include the scattering zippers with size $L=2l$.
In fact,  the notations for the functions $\varphi_\sharp$ and $\varphi_\flat$ below will be written in a form in which they generalize to the $l$-channel case.
Analogues of the theorems in this article for the $l$ channel case will be considered elsewhere.
\item Of course one can easily extend the definitions to `doubly' infinite one-channel operators using a direct sum over the whole integers, $\bigoplus\limits_{n \in \ZZ} \ell^2(\Sb_n)$. But such operators can be treated as a finite rank perturbation of a direct sum of two one-channel operators as defined above and for simplicity we omit a detailed discussion here. 
\end{enumerate}
\end{rem}

\subsection{Transfer matrices\label{sub-tr}} 

Let us start with the following proposition which also defines the maps $\varphi_\sharp$ and $\varphi_\flat$ that will be useful for describing the transfer matrices.

\begin{proposition}\label{prop-varphi}
For a matrix 
$$
M=\pmat{\alpha & \beta \\ \gamma & \delta}\,\in\,\CC^{2 \times 2},
$$
where $\beta \neq 0$, define
$$
\varphi_\sharp(M)\,=\,\pmat{\beta^{-1} & -\beta^{-1}\alpha \\ \delta \beta^{-1} & \gamma-\delta \beta^{-1} \alpha} \qtx{and}
\varphi_\flat(M)\,=\,\pmat{\gamma-\delta\beta^{-1}\alpha & \delta \beta^{-1} \\ -\beta^{-1} \alpha & \beta^{-1}}\;.
$$
Then, 
$$
\pmat{\Psi_- \\ \Psi_+} \,=\, M \pmat{\Phi_- \\ \Phi_+}
\quad\Leftrightarrow\quad
\pmat{\Phi_+ \\ \Psi_+}\,=\,\varphi_\sharp(M) \pmat{\Psi_- \\ \Phi_-}\quad \Leftrightarrow\quad
\pmat{\Psi_+ \\ \Phi_+}\,=\,\varphi_\flat(M) \pmat{\Phi_- \\ \Psi_-}\;.
$$
We note $\varphi_\sharp(M)=\varphi_\flat(M^{-1})$.
Moreover, if $M\in \UU(2)$, then $\varphi_\sharp(M), \varphi_\flat(M) \in \UU(1,1)$, where
$$
\UU(1,1)\,=\,\left\{T\in \CC^{2 \times 2}\,:\, T^*\pmat{1 & 0\\ 0 & -1} T=\pmat{1 & 0\\ 0 & -1} \right\}.
$$
\end{proposition}

The proof of this proposition (in fact for the more general $L$-channel case) is in the appendix (cf. Proposition~\ref{prop-A0}).
We note that for the map $\varphi_\flat$ this proposition coincides with \cite[Theorem~6]{Sa-Rel} and \cite[Proposition~2]{MSb} and it gives the relation between scattering and transfer matrices 
as in the scattering theory of electronic conduction as developed by Landauer, Imry and Büttiker \cite{Bu-1, Bu-2, Im, La-1,La-2}.
The map $\varphi_\flat$ also appears in the construction of transfer matrices for scattering zippers as in \cite{MSb},  whereas the map $\varphi_\sharp$ is used in the formulas for the reduced transfer matrices in the Hermitian case  \cite{DC, Sa-OC, SB} without giving it a symbol.
We also note that $\varphi_\flat$ gives the relation between scattering matrix and transfer matrix for Jacobi operators in an adequate bases as denoted in \cite[Theorem~6 \& Appendix]{Sa-Rel}.

In order to get to the transfer matrices we first re-write the eigenvalue equation as a system of equations, similar as in \cite{MSb}.
\begin{proposition}\label{prop-eig-equation}
The following set of equations are equivalent (in fact for solutions $\Psi, \Phi \in \CC^\GG$)
\begin{enumerate}[{\rm (i)}]
\item $\Uu^{(u)} \Psi = z \Psi\qtx{{\rm and}} \Ww^{(u)} \Phi\,=\,\Psi$.
\item $\Vv \Psi = z \Phi \qtx{{\rm and}} \Ww^{(u)} \Phi\,=\,\Psi$.
\item $\tilde \Uu^{(u)} \Phi = z \Phi\,\qtx{{\rm and}} \Ww^{(u)} \Phi\,=\,\Psi$.
\end{enumerate}
\end{proposition}
\begin{proof} For the proof we will omit the index $u$.
Using the fact that $\Ww$ is a finite hopping operator, we get natural extensions to $\CC^\GG$ and these extensions satisfy
$\Ww^* \Ww=\one$ with $\one$ being the identity operator on $\CC^\GG$. In particular, $\Ww$ is invertible as an operator on $\CC^\GG$.
So if $\Ww \Phi=\Psi$ or $\Ww^{-1} \Psi=\Phi$ we find
$$
\Uu\Psi=z\Psi \;\Leftrightarrow\; \Ww^{-1}\Uu\Psi=\Ww^{-1} z \Psi\; \Leftrightarrow \; \Vv \Psi=z\Phi\;\Leftrightarrow\;
\Vv\Ww\Phi\,=\,z \Phi \;\Leftrightarrow\;\tilde \Uu \Phi= z \Phi.
$$
\end{proof}

We will use the equations (ii) to define the transfer matrices. 
First note from $\Ww^{(u)}\Phi= \Psi$, one has
$$
\Psi_{(0,-)} = u \Phi_{(0,-)}\;,\quad P_n \Psi_n = P_n \Phi_n \qtx{and} \pmat{\Psi_{(n-1,+)} \\ \Psi_{(n,-)}}\,=\,W_n\,\pmat{\Phi_{(n-1,+)} \\ \Phi_{(n,-)}}.
$$
Then, $\Vv \Psi=z\Phi$ gives 
$$
(z^{-1}V_n-P_n) \Psi_n\,=\, z^{-1} V_n \Psi_n-P_n \Phi_n\,=\,Q_nQ_n^*\Phi_n\,=\,Q_n \pmat{\Phi_{(n,-)} \\ \Phi_{(n,+)}}\;,
$$
which implies
\begin{equation}\label{eq-eigenvalue-V}
\pmat{ \Psi_{(n,-)} \\ \Psi_{(n,+)}}\,=\, \pmat{\alpha_{z,n} & \beta_{z,n} \\ \gamma_{z,n} & \delta_{z,n}}\pmat{\Phi_{(n,-)} \\ \Phi_{(n,+)}},
\end{equation}
where
\begin{align}\label{eq-def-alpha-etc}
\pmat{\alpha_{z,n} & \beta_{z,n} \\ \gamma_{z,n} & \delta_{z,n}}\,&=\,
Q_n^* (z^{-1} V_n-P_n)^{-1} Q_n \nonumber \\
&= \pmat{e_{(n,-)}^* \\ e_{(n,+)}^*} (z^{-1}V_n-P_n)^{-1} \pmat{ e_{(n,-)} & e_{(n,+)} },
\end{align}
in case that $z^{-1}V_n-P_n$ is invertible. 
We also note that for $|z|=1$ the matrix defined in \eqref{eq-def-alpha-etc} is unitary by part b) of Proposition~\ref{prop-A1}, where
$$
A=Q_n^* z^{-1} V_n Q_n,\,B=Q_n^* z^{-1}V_n Q_n^\perp,\,C=(Q_n^\perp)^* z^{-1} V_n Q_n,\, D=(Q_n^\perp)^*z^{-1}V_n Q_n^\perp\,,$$ 
$$
P=(Q_n^\perp)^* P_n Q_n^\perp\,=\, (Q_n^\perp)^* Q_n^\perp\,=\,\one\,.$$ 
Here, the column vectors of $Q_n^\perp \in \CC^{\Sb_n \times (|\Sb_n|-2)}$  complete the columns of $Q_n$ to an orthonormal basis of $\CC^{\Sb_n}$.

\begin{rem} \label{rem-mer}
Clearly,  $z\mapsto Q_n^*(z^{-1}V_n-P_n)^{-1} Q_n$ is a rational function by Cramer's rule, it exists for all $0<|z|<1$. Using Proposition~\ref{prop-A1}, none of the poles lies on the unit circle. Rewriting $(z^{-1}V_n-P_n)^{-1}=z(V_n-zP_n)^{-1}$, it is easy to see
that in the limit $z\to 0$ one obtains the zero matrix. 
Hence, after analytic continuation, $\alpha_{z,n}, \,\beta_{z,n}, \,\gamma_{z,n}$ and $\delta_{z,n}$ are well defined for all $|z|\leq 1$.
\end{rem}

Furthermore,  note that equivalently one may derive
\begin{equation*}
\pmat{\Phi_{(n,-)} \\ \Phi_{(n,+)}}\,=\,\pmat{\tilde \alpha_{z,n} & \tilde \beta_{z,n} \\ \tilde \gamma_{z,n} & \tilde \delta_{z,n}}\,\pmat{ \Psi_{(n,-)} \\ \Psi_{(n,+)}},
\end{equation*}
where
\begin{equation*}
\pmat{\tilde \alpha_{z,n} & \tilde \beta_{z,n} \\ \tilde \gamma_{z,n} & \tilde \delta_{z,n}}\,=\,Q_n^*(zV_n^*-P_n)^{-1} Q_n.
\end{equation*}
This is well defined for all $|z|>1$.
In the case where all inverses exist, one thus gets
\begin{equation}\label{eq-rel-inverses}
\big(Q_n^*(z^{-1} V_n-P_n)^{-1} Q_n\big)^{-1}\,=\,Q_n^*(zV_n^*-P_n)^{-1} Q_n,
\end{equation}
which is a special case of Proposition~\ref{prop-A1}~c) and also shows unitarity for $|z|=1$.
The guide for defining the transfer matrices is the special case where all $W_n=\pmat{0&1\\1 & 0}$ as it will be the case for a 1D quantum walk, cf. subsection~\ref{sub:QW}.  

For this choice of $W_n$ one obtains
$\Psi_{(n+1,-)}= \Phi_{(n,+)}$, $\Psi_{(n-1,+)}=\Phi_{(n,-)}$,  and it makes sense to define a transfer matrix associated to $V_n$ by
\begin{equation}\label{eq-def-T-sharp}
\pmat{\Phi_{(n,+)} \\ \Psi_{(n,+)} }\,=\,T_{z,n}^\sharp \,\pmat{\Psi_{(n,-)} \\ \Phi_{(n,-}) }.
\end{equation}
By \eqref{eq-eigenvalue-V} and Proposition~\ref{prop-varphi} $T^\sharp_{z,n}$ exists if $\beta_{z,n} \neq 0$ (or $\tilde \beta_{z,n} \neq 0$),  in which case
\begin{equation}\label{eq-T-sharp-1}
T^\sharp_{z,n}\,=\,\varphi_\sharp\big(Q_n^*(z^{-1}V_n-P_n)^{-1} Q_n\big)\,=\,\pmat{\beta_{z,n}^{-1} & -\alpha_{z,n} \beta_{z,n}^{-1} \\ \delta_{z,n} \beta_{z,n}^{-1} & \gamma_{z,n} - \delta_{z,n} \beta_{z,n}^{-1} \alpha_{z,n} },
\end{equation}
or
\begin{equation}\label{eq-T-sharp-2}
T^\sharp_{z,n}\,=\,\varphi_\flat\big( Q_n^*(zV_n^*-P_n)^{-1}Q_n\big)
\pmat{\tilde \gamma_{z,n}-\tilde \delta_{z,n} \tilde \beta_{z,n}^{-1} \tilde \alpha_{z,n} & \tilde \delta_{z,n} \tilde \beta_{z,n}^{-1} \\ -\tilde \beta_{z,n}^{-1} \tilde \alpha_{z,n} & \tilde \beta_{z,n}^{-1}}.
\end{equation}

In order to complete to a transfer matrix coming from the `level before',  we define the transfer matrix associated to $W_n$ by
$$
\pmat{\Psi_{(n,-)} \\ \Phi_{(n,-)}}\,=\,T^\flat_n \pmat{\Phi_{(n-1,+)} \\ \Psi_{(n-1,+)}},
$$
and let
\begin{equation}\label{eq-rel-Tr}
T_{z,n}=T^\sharp_{z,n} T^\flat_n \qtx{to get} 
\pmat{\Phi_{(n,+)} \\ \Psi_{(n,+)}}\,=\,T_{z,n} \pmat{\Phi_{(n-1,+)} \\ \Psi_{(n-1,+)}}.
\end{equation}
By Proposition~\ref{prop-varphi} for $n\geq 1$,
\begin{equation}\label{eq-T-flat}
T^\flat_n\,=\, \varphi_\flat(W_n)= \pmat{c_n-d_n b_n^{-1} a_n & d_n b_n^{-1} \\ -b_n^{-1} a_n & b_n^{-1}}=\varphi_\sharp(W_n^*)\qtx{where} W_n=\pmat{a_n & b_n \\ c_n & d_n}.
\end{equation}
Note that for the special choice $W_n=\pmat{0&1\\1&0}$ we have $T_n^\flat=\one_2$.

The equation $\Psi_{(0,-)}=u\Phi_{(0,-)}$ for the operator pair $(\Uu^{(u)}, \tilde \Uu^{(u)})$ can be understood as some boundary condition.
Here,  we will not incorporate the boundary condition into the transfer matrices and simply define
$$
T^\flat_0\,=\,\one_2\,=\,\pmat{1 & 0\\ 0 & 1}\;, \quad T_{z,0}=T_{z,0}^\sharp,
$$
meaning that, formally, $\Phi_{(-1,+)}:=\Psi_{(0,-)}$ and $\Psi_{(-1,+)}:=\Phi_{(0,-)}$.
Then, the boundary condition becomes
$
\Phi_{(-1,+)}\,=\,u \Psi_{(-1,+)}.
$
Now, in order that transfer matrices exist, we assume the following.
\\[.2cm]
\vbox{
{\bf Assumptions}
\begin{enumerate}
\item[(A1)] For all $n\geq 0$, there exists $k\in\NN$ such that $e_{(n,+)}^* V_n^k e_{(n,-)} \neq 0$. This simply  means that $V$ connects the backwards moving mode $e_{(n,-)}$ of the $n$-th shell to its forward moving mode $e_{(n,+)}$.
\item[(A2)] For all $n \geq 1$,  $0\neq b_n=  e_{(n-1,+)}^* \Ww e_{(n,-)}=\pmat{1&0} W_n \pmat{1\\0} $.
\end{enumerate}}

\begin{proposition} \label{prop-inf-dim}
{\rm (A1)} and {\rm (A2)} are both fulfilled if and only if the $\Uu^{(u)}$-cyclic space generated by  $e_{(0,-)}$ is infinite dimensional.
\end{proposition}
\begin{proof}
We write $\Uu$ for $\Uu^{(u)}$. 
If (A2) is not satisfied at the level $n$, then there is no connection from $\Sb_{n-1}$ to $\Sb_n$ and
both, $\Vv$ and $\Ww$ , and thus $\Uu$, leave the space $\bigoplus_{m=0}^{n-1} \ell^2(\Sb_m)$ invariant.
Therefore, the cyclic space generated by $e_{(0,-)}$ is finite dimensional. \\
If (A1) is not satisfied at the level $n$, then $\ell^2(\Sb_n)=\HH_{n,-} \oplus \HH_{n,+}$ where $V_n$ leaves both spaces invariant and $e_{(n,-)}\in\HH_{n,-}$, $e_{(n,+)} \in \HH_{n,+}$. Therefore, $\Vv$ and $\Ww$, and thus $\Uu$, leave the space
$\bigoplus_{m=0}^{n-1} \ell^2(\Sb_m) \oplus \HH_{n,-}$ invariant, and the cyclic space of $e_{(0,-)}$ is finite dimensional.\\[.2cm]
Now assume (A1) and (A2) are both fulfilled. For $k$ where $V_0^k e_{(0,-)}$ is perpendicular to $e_{(0,+)}$, we have $\Uu^k e_{(0,-)} \subset \ell^2(\Sb_0)$.
For the minimum $k_1$ where $V_0^{k_1} e_{(0,-)}$ has overlap with $e_{(0,+)}$, $\Ww$ will transport to the next shell. Then, $k_1$ is also the minimum
such that $\Uu^{k_1} e_{(0,-)}$ is not orthogonal to $e_{(1,-)}$.
Repeating the arguments and following the quantum walk along the shells, we find a sequence $k_1 < k_2 < \ldots$ of positive integers, 
such that $k_n$ is the minimum number where $\Uu^{k_n} e_{(0,-)}$ is not orthogonal to $e_{(n,-)}$.
Clearly, $\Uu^{k_n} e_{(0,-)}$ are all linearly independent and the cyclic space generated by $e_{(0,-)}$ is infinite dimensional.
\end{proof}

Note that the splitting of the operator into a direct sum when (A1) or (A2) are invalidated, shows that there should be no transfer. We may speak of a `broken channel' in this case. 
On the other hand, when they are fulfilled we generate the infinite dimensional cyclic space and there should always be a transfer.
Indeed, we immediately see that (A2) is equivalent to all transfer matrices $T^\flat_n$ being well defined.
Considering the assumption (A1) we have the following equivalence.

\begin{proposition}
The following properites  are equivalent
\begin{enumerate}[{\rm (i)}]
\item Assumption {\rm (A1)}.
\item For all $n$,  $z\mapsto \beta_{z,n}$ is not the zero function.
\item For all $n\in\NN$, $z \mapsto \tilde \beta_{z,n}$ is not the zero function.
\end{enumerate}
Thus, in this case,  for any $n$,  $T^\sharp_{z,n}$ is well defined except for finitely many $z$, as $z\mapsto \beta_{z,n}$ is a rational function.
\end{proposition}
\begin{proof}
First note that it is easy to see that $z\mapsto \beta_{z,n}$ is the zero function if and only if $z\mapsto \tilde \beta_{z,n}$ is the zero function.  
Now, for $|z|<1$ we have by \eqref{eq-def-alpha-etc}
 $$\beta_{z,n}= e_{(n,-)}^* (\one -zV_n^* P_n )^{-1}zV_n^*  e_{(n,+)} = \sum_{k=0}^\infty z^{k+1} e_{(n,-)}^*(P_nV_n^*)^k V_n^* e_{(n,+)},
 $$ which means that
$z\mapsto  \beta_{z,n} $ is identically zero if and only if 
for all $k\in \NN_0$,
$$e_{(n,+)}^* V_n(P_nV_n)^k e_{(n,-)}=0.$$
To finish the proof it suffices to prove now by induction in $m$ that\\[.2cm]
 $e_{(n,+)}^* V_n(P_nV_n)^k e_{(n,-)}=0$ for all $k=0,\ldots,m$ if and only if $e_{(n,+)}^* V_n^{k+1} e_{(n,-)}=0$ for all $k=0,\ldots,m$. \\[.2cm]
The case $ m=0$ is clear. For the induction step $m-1\to m$ let $R_n=1-P_n=e_{(n,-)} e_{(n,-)}^* + e_{(n,+)} e_{(n,+)}^*$ and note
that in either direction, by hypothesis and induction hypothesis we have for $l=1,\ldots,m$
$$
e_{(n,+)}^*V_n^l R_n V_n (P_n V_n)^{m-l}  e_{(n,-)}=e_{(n,+)}^* V_n^l e_{(n,+)} \underbrace{e_{(n,+)}^* V_n(P_n V_n)^{m-l} e_{(n,-)}}_{=0}
$$
$$
\qquad+\;\underbrace{e_{(n,+)}^* V_n^l e_{(n,-)}}_{=0} e_{(n,-)}^*V_n (P_n V_n)^{m-l}  e_{(n,-)}\,=\,0\;.
$$
Now,
$$
V_n^{m+1} \,=\, V_n^m R_n V_n+V_n^m P_nV_n\,=\,V_n^{m} R_n V_n + V_n^{m-1} R_n V_n P_n V_n + V_n^{m-1} (P_n V_n)^2
$$
$$
=\ldots= \sum_{l=1}^m V_n^l R_n (V_n P_n)^{m-l} V_n + V_n (P_n V_n)^m,
$$
and the previous statement gives $e_{(n,+)}^* V_n^{m+1} e_{(n,-)}=0$ if and only if $ e_{(n,+)}^*V_n (P_n V_n)^m e_{(n,-)}=0$. which finishes the induction step as $e_{(n,-)}^* V_n^l e_{(n,+)}=0$ for $l=0,\ldots, m$ by induction hipothesis.
\end{proof}

On the set where either $(z^{-1}V_n-P_n)^{-1}$ or $\varphi_\sharp\big(Q_n^*(z^{-1}V_n-P_n)^{-1}Q_n\big) $ is not defined we may use analytic continuation in $z$ to define $T^\sharp_{z,n}$ wherever possible.   
By Remark~\ref{rem-mer} one can use the formula \eqref{eq-T-sharp-1}  for $0<|z|\leq 1$ where $\beta_{z,n} \neq 0$,  and one can use 
\eqref{eq-T-sharp-2} for $|z|>1$ where $\tilde \beta_{z,n} \neq 0$. The exceptional set where $T^\sharp_{z,n}$ and $T_{z,n}$ are not defined is thus given by
$$
\widehat \Aa_n\,=\,\{z\,:\,0< |z| \leq 1\,,  \beta_{z,n}=0\}\,\cup\,\{z\,:\,|z|>1\,,\, \tilde \beta_{z,n}=0 \}\;.
$$
More important for spectral theory are the sets
\begin{equation}\label{eq-Aa_n}
\Aa_N\,=\,\bigcup_{n=0}^N \widehat\Aa_n \cap \UU(1) = \{z\,:\,\big(|z|=1\,\wedge\,\exists n,\,  0\leq n\leq N\,:\,\beta_{z,n}=0\big) \}
\end{equation}
\begin{equation}\label{eq-Aa}
\qtx{and} \Aa=\bigcup_{n=0}^\infty \Aa_n.
\end{equation}
Note that $\Aa_N$ are finite sets under assumption (A2)  and $\Aa$ is thus countable.
\newline
Apart from the set where some transfer matrices are not defined, we can consider the products
\begin{equation}\label{eq-product-T}
T_{z,[0,n]}\,:=\,T_{z,n} T_{z,n-1} \cdots T_{z,1} T_{z,0}\;.
\end{equation}

\subsection{Spectral average formula and criteria for a.c. spectrum\label{sub-sp-av}}

Let $\mu^{(u)}$ denote the spectral measure of  the operator $\Uu^{(u)}$ (or alternatively $\tilde \Uu^{(u)}$) at the vector $e_{(0,-)}$,  meaning
$$
\mu^{(u)}(f)\,=\,e_{(0,-)}^*\, f(\Uu^{(u)})\, e_{(0,-)}\,=\,e_{(0,-)}^*\, f(\tilde \Uu^{(u)})\, e_{(0,-)}\,.
$$
In Dirac notation we would write
$$
\mu^{(u)}(f)\,=\,\langle e_{(0,-)}\,|\, f(\Uu^{(u)})\,|\, e_{(0,-)}\rangle \,=\,\langle e_{(0,-)}\,|\, f(\tilde \Uu^{(u)})\, |\, e_{(0,-)}\,\rangle\;.
$$
The second equation follows easily as $\tilde \Uu^{(u)} = (\Ww^{(u)})^* \Uu^{(u)} \,\Ww^{(u)}$ and $\Ww^{(u)} e_{(0,-)}=u\,e_{(0,-)}$ with $|u|=1$.

Similarly to the Hermitian case (cf. \cite{Sa-OC})
we need to separate the measure part induced by compactly supported eigenfunctions. First we define
\begin{equation} \label{eq-def-HH_c}
\HH^{(u)}_c\,=\,\overline{ {\rm span}\,\{\psi\in\ell^2(\GG)\,:\, \psi\,\text{compactly supported eigenfunctions of $\Uu^{(u)}$ }\,\} },
\end{equation}
where the bar denotes the closure. 
Then, let $P^{(u)}$ be the orthogonal projection onto $\HH^{(u)}_c$ and define the point measure
\begin{equation} \label{eq-def-nu^u}
\nu^{(u)}(f)\,=\, \langle P^{(u)} e_{(0,-)}\,|\, f(\Uu^{(u)})\,|P^{(u)} e_{(0,-)}\,\rangle\;.
\end{equation}

\begin{rem}
Note that for some particular eigenvalue $z_0\in\UU(1)$, the eigenfunction in the intersection of the cyclic space of $e_{(0,-)}$ with $\HH_c^{(u)}$ giving $\nu^{(u)}(f)$ may not be compactly supported, but in such a case, it must be the limit of compactly supported eigenfunctions of the same eigenvalue. This only may happen if $z_0$ is an eigenvalue of infinite multiplicity.
\end{rem}

We obtain the following analogue to \cite[Theorem~2]{Sa-OC}.

\begin{Theo}\label{th-main}
For any $u \in \UU(1)$,  $\nu^{(u)}$ is supported on $\Aa$ and for $f\in C(\partial \DD)=C(\UU(1))$ we have
\begin{align*}
\mu^{(u)}(f)\,&=\,\nu^{(u)}(f)\,+\,\lim_{n \to \infty}
\int_{0}^{2 \pi} \frac{f(e^{i\varphi})}{\pi} \,\frac{{\rm d}\varphi}{\left\|T_{e^{i\varphi},[0,n]} \smat{u\\1} \right\|^2}.
\end{align*}
Note that this can be interpreted as some weak limit convergence of the absolute continuous measures on the unit disk with densities $\pi^{-1} \left\|T_{e^{i\varphi},[0,n]} \smat{u\\1} \right\|^{-2}$ towards $\mu^{(u)}-\nu^{(u)}$.
\end{Theo}

We can deduce Carmona's criterion for one-channel operators, a generalization of \cite[Theorem~10.7.5]{Si-OPUC} in the CMV case.

\begin{Theo}\label{th-main2}
Assume that for $p>1$, and 
$\varphi_0< \varphi_1$ one has
$$
\liminf_{n\to \infty} \int_{\varphi_0}^{\varphi_1} \|T_{e^{i\varphi},[0,n]} \|^{2p}\,{\rm d}\varphi\,<\,\infty,
$$
then, for any $u\in\UU(1)$, the positive measure $\mu^{(u)}-\nu^{(u)}$ is purely absolutely continuous in $e^{i(\varphi_0,\varphi_1)}=\{e^{i\varphi}\,:\,\varphi\in (\varphi_0,\varphi_1)\}$ w.r.t.  the Haar measure on $\partial\DD=\UU(1)$, has density in $L^p(e^{i(\varphi_0,\varphi_1)})$, and 
$$e^{i[\varphi_0,\varphi_1]} \subset \supp (\mu^{(u)}-\nu^{(u)}) \;,\qtx{in particular} e^{i[\varphi_0,\varphi_1]} \subset \sigma_{ac}(\Uu^{(u)})\,. $$
\end{Theo}

Both theorems will be proved in Section~\ref{sec-sp-av}.


\subsection{Absolutely continuous spectrum for periodic one-channel scattering zippers with random decaying perturbation\label{sub-random}}

For a scattering zipper $\Uu=\Ww\Vv$ as in \cite{MSb} we have $|\Sb_n|=2$ for all shells and thus $V_n, W_n \in \UU(2)$.
In particular, $e_{(n,-)}, e_{(n,+)}$ is a basis of $\CC^{\Sb_n}$ and $P_n=\nul$ and combining them to a basis of $\ell^2(\GG)$ we have
$$
\Vv=\pmat{V_0 \\  & V_1 \\. & & V_2 \\ & & & \ddots}\,, \quad \Ww=\pmat{u \\ & W_1 \\ & & W_2 \\ & & & \ddots}.
$$
Furthermore, we will write
\begin{equation}\label{eq-VnWn}
V_n\,=\,\pmat{\aaa_n & \bbb_n \\ \ccc_n & \ddd_n} \in \UU(2)  \qtx{and}
W_n\,=\,\pmat{a_n & b_n \\ c_n & d_n} \in \UU(2).
\end{equation}
Then, assumption (A1)  reduces to $\ccc_n\neq 0$, or equivalently $\bbb_n\neq 0$, for all $n\in\ZZ_+$, similar as assumption (A2).
Thus, let us assume $b_n\neq 0, \bbb_n\neq 0$ for all $n$.
In this case, all transfer matrices are defined for all $z\in \CC^*$ and we have
\begin{equation}\label{eq-T-sharp-SZ}
T_{z,n}^\sharp\,=\,\varphi_\sharp(zV_n^{-1})\,=\,\varphi_\flat(z^{-1} V_n)\,=\,\pmat{z^{-1}(\ccc_n-\ddd_n \bbb_n^{-1} \aaa_n) & \ddd_n \bbb_n^{-1} \\ -\bbb_n^{-1} \aaa_n & z \bbb_n^{-1}}\,,
\end{equation}
\begin{equation}\label{eq-T-flat-SZ}
T_{n}^\flat\,=\,\pmat{c_n-d_n b_n^{-1} a_n & d_n b_n^{-1} \\ -b_n^{-1} a_n & b_n^{-1}}\qtx{and} T_{z,n}=\begin{cases} T_{z,n}^\sharp T_n^\flat & \text{if}\,\, n\geq 1 \\ T_{z,0}^\sharp & \text{if}\,\, n=0.
\end{cases}
\end{equation}

We say that the scattering zipper is $p$-periodic if 
$$
V_{n+p}\,=\,V_n \qtx{for all $n\in\ZZ_+$, \quad and}
W_{n+p}\,=\,W_n \quad \text{for all $n\in\ZZ_+^*$}.
$$

Note that under these conditions $T_{z,n+p}=T_{z,n}$ for any $z\in \CC^*$ and all $n\geq 1$. Moreover, we have
$$T_{z,0}=T_{z,0}^\sharp = T_{z,p}^\sharp = T_{z,p} (T_p^\flat)^{-1}.$$
Then, we define the transfer matrix over one period $p$ by $T_z=T_{z,[1,p]}$ and we have 
\begin{equation}
T_{z,[0,np]}\,=\,T_z^n T_{z,0}.
\end{equation}
For $z \in \UU(1)$ we have $T_z \in \UU(1,1)$, and therefore, one obtains 
$$\frac{(\Tr T_z)^2}{ \det(T_z)} \geq 0\quad \text{and}\quad |\det T_z|=1.$$
This means, one has eigenvalues of the form $e^{i\chi} \lambda$ and $e^{i\chi}\lambda^{-1}$ where $\lambda+\lambda^{-1} \in \RR$ and $e^{2i\chi}=\det(T_z)$ (see Proposition~\ref{prop-U11}).

If $|\Tr T_z|<2$, we find that $\lambda \in \UU(1)$ and $\|T^n_z\|$ is uniformly bounded. Thus one can use Theorem~\ref{th-main2} to find that
there is absolutely continuous spectrum.
Hence, we define the set
\begin{equation}\label{eq-sigma-ac}
\Sigma\,=\, \{z\in \UU(1)\,:\,  |\Tr T_z |\,<\,2\,\}\,=\,\Big\{z \in \UU(1)\,:\, \frac{(\Tr T_z)^2}{\det(T_z)} < 4\,\Big\}.
\end{equation}

\begin{proposition}\label{prop-sigma-ac}
The set $\Sigma$ is a non-empty union of open intervals on $\UU(1)$.
Apart from a finite set of eigenvalues outside $\overline{\Sigma}$, the spectrum of $\Uu$ is purely absolutely continuous and
given by the closure of $\Sigma$. 
$$\sigma_{ess}(\Uu)\,=\,\sigma_{ac}(\Uu)\,=\,\overline{\Sigma}.$$
\end{proposition}

Now we consider random $\ell^2$ perturbations of $(V_n)_n$. This means, let $(\Omega,\Aa,\PP)$ be a probability space and let $\widehat V_n,\,\widehat W_n \,:\,\Omega \to \UU(2)$, be unitary matrix valued random variables and 
consider the random perturbations where $W_n$ and $V_n$ are replaced by $\widehat W_n$ and $\widehat V_n$ so that
\begin{equation}\label{eq-Uu_omega}
\Uu_\omega\,=\, \Ww_\omega\,\Vv_\omega \qtx{where} \Vv_\omega\,=\, \bigoplus_{n=0}^\infty \widehat V_n(\omega)\;,\quad
\Ww_\omega\,=\,u\oplus \bigoplus_{n=1}^\infty \widehat W_n.
\end{equation}
As in \eqref{eq-VnWn} we define the entries 
$$
\widehat V_n\,=\,\pmat{\hat \aaa_n & \hat \bbb_n \\ \hat \ccc_n & \hat \ddd_n} \in \UU(2)  \qtx{and}
\widehat W_n\,=\,\pmat{\hat a_n & \hat b_n \\ \hat c_n & \hat d_n} \in \UU(2),
$$
and the corresponding transfer matrices will be denoted by 
$$
\widehat T_{z,n}^\sharp\,=\,\pmat{z^{-1}(\hat \ccc_n-\hat \ddd_n \hat \bbb_n^{-1} \hat \aaa_n) & \hat \ddd_n \hat \bbb_n^{-1} \\ -\hat \bbb_n^{-1} \hat \aaa_n & z \hat \bbb_n^{-1}}\,,\quad
\widehat T_{n}^\flat\,=\,\pmat{\hat c_n-\hat d_n \hat b_n^{-1} \hat a_n & \hat d_n \hat b_n^{-1} \\ -\hat b_n^{-1} \hat a_n & \hat b_n^{-1}},
$$
and 
$$
\widehat T_{z,0}\,=\,\widehat T_{z,0}^\sharp\,, \quad \widehat T_{z,n}\,=\,\widehat T_{z,n}^\sharp \widehat T_{n}^\flat \qtx{for} n\geq 1\,.
$$
 We assume that  the following conditions hold.\\[.2cm]
(C1) The family of pairs $\displaystyle \{ (\widehat V_n, \widehat W_n)\}_{n=0}^\infty$ is independent\footnote{equivalently, one may state that the family of transfer matrices $(\widehat T_{z,n})_n$ is independent, $\widehat V_n$ and $\widehat W_n$ may have correlations. } . \\[.2cm]
(C2)  $\displaystyle \sum_{n=1}^\infty \left( \big\| \EE\big(\widehat V_n\big) - V_n \big\| + \EE\big(\|\widehat V_n-V_n\|^2\big) + 
\big\|\EE\big(\widehat W_n\big) - W_n \big\| + \EE\big(\|\widehat W_n-W_n\|^2\big) \right)\,<\,\infty$.\\[.2cm]
(C3) $\exists\, \varepsilon>0,\; \forall n\in\ZZ_+\,:\, |\hat b_n|>\varepsilon\,\wedge\, |\hat \bbb_n|>\varepsilon$ almost surely.

\vspace{.3cm}
Note, formally $\widehat W_0$ does not appear in the operator, so one may define it as some deterministic matrix for the purpose of assumption (C1).
Assumption (C2) makes sure that the perturbation is Hilbert Schmidt. 
We get the following analogue to Kiselev-Last-Simon's result for decaying potentials on the line, \cite[Theorem~8.1]{KiLS}.
\begin{Theo}\label{th-main3}
Assume that {\rm (C1)}, {\rm (C2)}, {\rm (C3)} hold.
Then, there is a set $\Omega'$  of probability one, $\PP(\Omega')=1$, such that for all $\omega \in \Omega'$,
the spectrum of $\Uu_\omega$ is purely absolutely continuous in $\Sigma$, and, $\sigma_{ac}(\Uu_\omega)\,=\,\overline{\Sigma}=\sigma_{ess}(\Uu_\omega)$.
\end{Theo}

\begin{rem}
\mbox{}
\begin{enumerate}[{\rm (i)}]
\item In order to quickly adopt the simple argument from Kiselev Last Simon \cite{KiLS}  one needs the technical 
condition
$$
\sum_{n=0}^\infty 
\Big( \big\| \EE(\Delta T_{z,n} \big\|\,+\, \EE\big( \|\Delta T_{z,n}\|^2\,+\,\|\Delta T_{z,n}\|^4\,\big)\Big)
\,<\infty\,.
$$
where
$$
\Delta T_{z,n}\,=\, \widehat T_{z,[np+1,(n+1)p]}- T_{z,[np+1,(n+1)p]} \,,
$$
for any $z\in \Sigma$, uniformly for $z$ in compact subsets of $\Sigma$. Assumptions (C1), (C2), (C3) guarantee this.
Without assumption (C3) one could have $\hat b_n$ or $\hat \bbb_n$ closer and closer to zero with smaller and smaller positive probabilities, such that (C2) is satisfied but not (C3) and also not the technical condition needed as stated above, as the inverses of $\hat b_n, \hat \bbb_n$ appear in the transfer matrices.

\item Using techniques from \cite{GS} one can get rid of assumption (C3) with probabilistic arguments.
For an analogue of Theorem~\ref{th-main3} for general one-channel operators the corresponding condition (C3)  would be very technical. Again, using techniques from \cite{GS}  assumptions (C1), (C2) are sufficient. However, the proof  would be much more technical and will be dealt with elsewhere.

\item  Theorem~\ref{th-main3} and its proof also work for deterministic perturbations where $\widehat V_n=\EE(\widehat V_n)$, $\widehat W_n=\EE(\widehat W_n)$.
But in this situation, condition (C2) actually states that the perturbation is trace class.

\item In other papers for the Hermitian case one typical has the assumptions like $\EE(\widehat V_n)=V_n,\,\EE(\widehat W_n)=W_n$. However, here, due to the formuals for the transfer matrices, such a condition would not imply $\EE(\widehat T_{z,n})=T_{z,n}$. Therefore, allowing different expectations and adjusting the condition as in (C2) makes no difference in the proofs.
\end{enumerate}
\end{rem}

\section{Examples \label{sec-examples} } 

\subsection{One-dimensional quantum walks\label{sub:QW}}

Typically, a one dimensional quantum walk is a unitary operator defined on the Hilbert space $\HH=\ell^2(\ZZ) \otimes \CC^2 \cong \ell^2(\ZZ \times \{ \uparrow, \downarrow\})$, 
given by a product $\Uu=\Ss\Cc$, where $\Cc$ is a direct sum of coins,
$\displaystyle\Cc=\bigoplus_{n\in\ZZ} C_n$, $C_n \in \UU(2)$  and $\Ss$ shifts spin ups forward, $\Ss \delta_{(n,\uparrow)}=\delta_{(n+1,\uparrow)}$ and spin downs backward, 
$\Ss \delta_{(n,\downarrow)}=\delta_{(n-1,\downarrow)}$.
This means, for $\psi=(\psi_n)_n \in \HH$, $\psi_n=\smat{\psi_{n,\uparrow} \\ \psi_{n,\downarrow}}$ one has
$$
(\Cc  \psi)_n\,=\,C_n \psi_n\;, \qquad (\Ss \psi)_n\,=\, \pmat{\psi_{n-1,\uparrow} \\ \psi_{n+1,\downarrow}}.
$$
Note that $\Cc, \Ss \in \UU(\HH)$.
For a half-line version on $\HH_+=\ell^2(\ZZ_+\times\{\uparrow,\downarrow\})$ one may change the definition of $\Ss$ slightly at $n=0$ by $(\Ss \psi)_0=\smat{\psi_{0,\downarrow} \\ \psi_{1,\downarrow}}$.

\begin{figure}[h!]
\begin{center}
\includegraphics[scale=0.35]{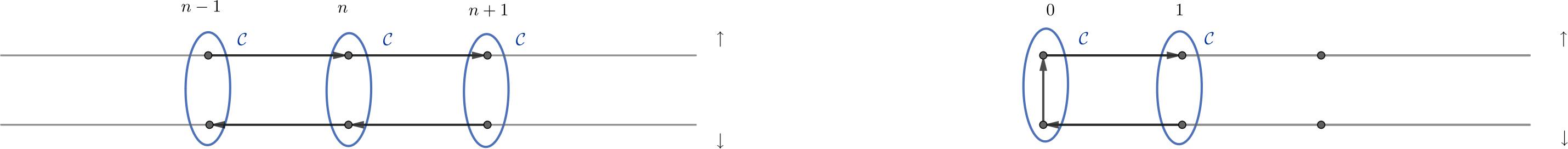}
\vspace*{-0.2cm}
\caption{The action of the shift  $\Ss$  and the coin operator $\Cc$ on $\HH$ and $\HH_+$.} \label{fig1}
\end{center}
\end{figure}

The half-line version can be transferred to our setup described above using the shells $\Sb_n= \{(n,\uparrow), (n,\downarrow)\}$. First, we define an adequate operator $\Ww$ by
\begin{equation}\label{eq-def-Ww-QW}
(\Ww \psi)_n \,=\, \pmat{ \psi_{n-1,\downarrow} \\ \psi_{n+1,\uparrow}} \qtx{for $n\geq 1$,} 
(\Ww \psi)_0\,=\, \pmat{\psi_{0,\uparrow} \\ \psi_{1,\uparrow}},
\end{equation}
then
\begin{equation*}
\Ww^2=I\;, \quad \Ww \Ss\,=\, \bigoplus_{n=0}^\infty \pmat{0 & 1 \\ 1 & 0}.
\end{equation*}

\begin{figure}[h!]
\begin{center}
\includegraphics[scale=0.25]{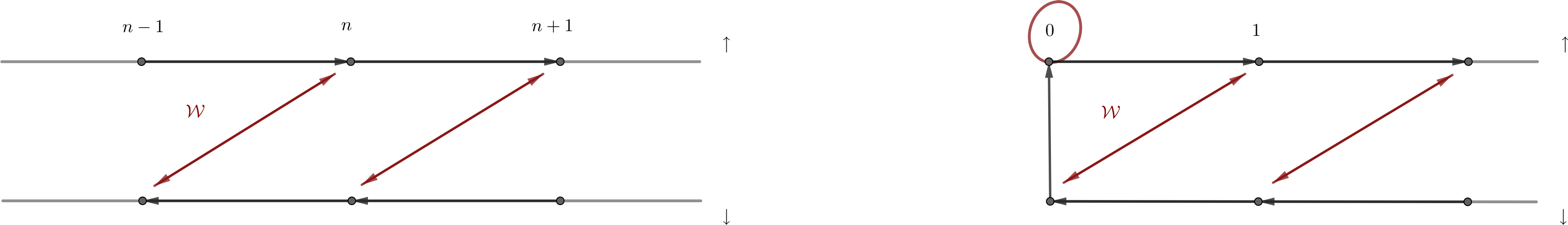}
\vspace*{-0.2cm}
\caption{The action of  $\Ww$   on $\HH$ and $\HH_+$.} \label{fig2}
\end{center}
\end{figure}

\begin{figure}[h!]
\begin{center}
\includegraphics[scale=0.25]{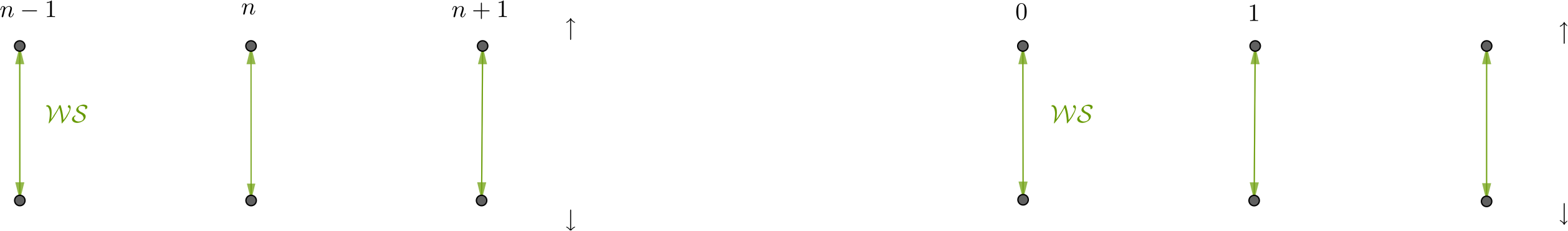}
\vspace*{-0.2cm}
\caption{The action of  $\Ww\Ss$   on $\HH$ and $\HH_+$.} \label{fig3}
\end{center}
\end{figure}

Therefore, defining 
$$\Vv= \Ww \Ss \Cc \qtx{and} V_n = \pmat{0 & 1 \\ 1 & 0} C_n\, \in\, \UU(2), $$ we find
\begin{equation}\label{eq-form-Uu-QW}
\Uu\,=\, \Ss\Cc\,=\, \Ww^2 \Ss \Cc\,=\, \Ww\,\Vv \qtx{where} \Vv=\bigoplus_{n=0}^\infty  V_n.
\end{equation}

Note, $\Ww$ interchanges $\delta_{(n,\downarrow)}$ with $\delta_{(n+1, \uparrow)}$, therefore, with
$e_{(n,+)}=\delta_{(n,\downarrow)}$, $e_{(n,-)}= \delta_{(n,\uparrow)}$ we have the structure as above in \eqref{eq-V}, \eqref{eq-W}, \eqref{eq-U} (see figure~\ref{fig3}) with 
\begin{equation*}
W_n=\pmat{0 & 1 \\ 1 & 0} \qtx{and} Q_n=\one \qtx{implying} P_n=\nul.
\end{equation*}
Using \eqref{eq-T-flat} this leads to
$$
T_n^\flat=\varphi_\flat(W_n)=\one\, \qtx{and} \Phi_{(n,+)}=\Psi_{(n+1,-)}=\Psi_{n+1,\uparrow}.
$$
Therefore, using  \eqref{eq-T-sharp-1}, \eqref{eq-T-sharp-2} and \eqref{eq-rel-Tr} we find for $\Uu \Psi= z \Psi$ the transfer matrix relation
$$
\pmat{\Psi_{n+1,\uparrow} \\ \Psi_{n,\downarrow}}=T_{z,n} \pmat{\Psi_{n,\uparrow} \\ \Psi_{n-1,\downarrow}},
$$
with
$$
T_{z,n}=\varphi_\flat(z^{-1} V_n)\,=\,\frac{1}{\overline{r_n}} \pmat{z^{-1} \omega_n & t_n \\ \overline{t_n} & z \overline{\omega_n}}\qtx{where} C_n=\omega_n \;\pmat{r_n&t_n \\ -\overline{t_n}& \overline{r_n} }\,,
$$
$|\omega_n|=1$ and $|t_n|^2+|r_n|^2=1$.

\subsection{Generalized one-channel quantum walks}

Let us now define generalized one-channel quantum walks.
As in the previous case, we have the general partition $$\displaystyle\HH=\ell^2(\GG)=\bigoplus_{n=0}^\infty \ell^2(\Sb_n),\quad 2\leq |\Sb_n|<\infty.$$ 
Within each $\Sb_n$ we assign some `spin up' and `spin down' orbitals, \hbox{$(n,\uparrow)$}, \hbox{$(n,\downarrow)$} $\in \Sb_n$ which are distinct. 
Note that, despite using notaitons `spin up' and `spin down', $\Sb_n$ may have many more orbitals. 
In order to define the quantum walk we just pick two for each $n$. 
As we want to show the analogue structure with the quantum walk and essentially use the same definition for the operators $\Ss$ and $\Ww$,
we use the same notation here.\\
We have a coin operator
$$
\Cc\,=\,\bigoplus_{n=0}^\infty C_n \qtx{where}
C_n \in \UU(\Sb_n),
$$
and a shift operator $\Ss$ defined by
$$
\Ss \delta_{(n,\uparrow)} = \delta_{(n+1,\uparrow)} \qtx{for all $n \in \ZZ_+$ and} \Ss \delta_{(n,\downarrow)} = \begin{cases} \delta_{(n-1,\downarrow)} & \text{for}\,\ n\geq 1 \\ 
\delta_{(0,\uparrow)} & \text{for}\,\ n=0,\end{cases}
$$
and in the orthogonal complement of the span of the $\delta_{(n,\uparrow)}, \delta_{(n,\downarrow)}$, $\Ss$ acts as identity.
\bigskip

Of course one may define a one-channel analogue of the $\ZZ$-quantum walk on the full line, where the extra case $n=0$ in the definition of $\Ss$ is not needed. 
 As before, the quantum walk is given by the unitary operator $\Uu=\Ss\Cc$.
As above, we define $\Ww$ as in \eqref{eq-def-Ww-QW} being the operator that interchanges $e_{(n,\downarrow)}$ with $e_{(n+1,\uparrow)}$ for $n\in\ZZ_+$, and $\Ww$  acts as identity on the orthogonal complement. Then, as in \eqref{eq-form-Uu-QW} we have a one-channel operator where
$e_{(n,+)}=\delta_{(n,\downarrow)}$, $e_{(n,-)}= \delta_{(n,\uparrow)}$,
$$
\Uu=\Ww \Vv \qtx{with} \Vv=\Ww \Ss \Cc = \bigoplus_{n=0}^\infty V_n \qtx{and} V_n=S_n C_n\;.
$$
Here, $S_n\in \UU(\Sb_n)$ is the unitary operator which interchanges $\delta_{(n,\downarrow)}$ and $\delta_{(n,\uparrow)}$ and acts as identity on the orthogonal complement. Note that here as well we have $W_n=\smat{0&1\\1&0}$  and we get
$$
\pmat{\Psi_{n+1,\uparrow} \\ \Psi_{n,\downarrow}}=T_{z,n} \pmat{\Psi_{n,\uparrow} \\ \Psi_{n-1,\downarrow}},
$$
where
$$
T_{z,n}\,=\,\varphi_\sharp\big(Q_n^*(z^{-1}V_n-P_n)^{-1} Q_n\big),
$$
and
$$
Q_n= \pmat{\delta_{(n,\uparrow)} & \delta_{(n,\downarrow)}},\; 
P_n=I_{\Sb_n}-  Q_n Q_n^* = I_{\Sb_n}- |\delta_{(n,\uparrow)} \rangle\langle \delta_{(n,\uparrow)}|\,-\,
|\delta_{(n,\downarrow)} \rangle\langle \delta_{(n,\downarrow)}|.
$$
A particular example of this kind could be a quantum walk on an infinite carbon chain \cite{PC-Carbon} where 
$\ell^2(\Sb_n)$ corresponds to the valence electron states of the carbon atom at position $n$ and possibly other atoms connected to it, $\delta_{(n,\uparrow)}, \delta_{(n,\downarrow)}$
are some orbitals that walk forward or backward. \\

\begin{figure}[ht!]
\begin{center}
\includegraphics[scale=0.5]{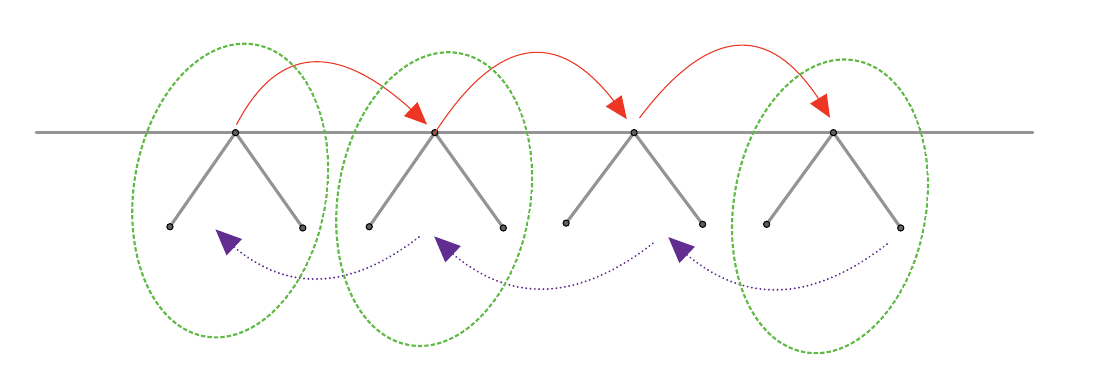}
\caption{A quantum walk on a carbon chain: The vertices on the line above are carbon atoms and the connected ones below may be for example hydrogen atoms. The `shells' correspond to the valence electron states of the groups of atoms indicated by the circles. In each group, one state shifts to the right, and another one to the left.} \label{fig4}
\end{center}
\end{figure}

\subsection{Stroboscopic unitary dynamics on ${\mathbb Z}^2$}

We consider a product of two unitary operators of the form $\Uu={\mathcal W}{\mathcal V}$ acting on $\ell^2 ({\mathbb Z}^2)$, where ${\mathcal V}$ is a configuration of the Chalker-Coddington model \cite{ABJ1, ABJ2} displaying a shell structure and ${\mathcal W}$ allowing transitions between the shells. 
Again, we consider the partition of ${\mathbb Z}^2$ in shells
  $$\displaystyle\HH=\bigoplus_{n=0}^\infty \ell^2(\Sb_n),$$
with a special configuration 
\begin{align*}
{\mathbb S}_n &= \big\{ (j_1,j_2) \in {\mathbb Z}^2 ; \| (j_1 , j_2) -(-\tfrac12, -\tfrac12) \|_{\infty} = n+\tfrac12 \big\},\quad |{\mathbb S}_n |= 4(2n+1) .
\end{align*}
Here $\| j\|_\infty=\max\{|j_1|,|j_2|\}$ for $j=(j_1,j_2)\in\ZZ^2$.
As above, let $(\Pp_{\Sb_n})_{n\in {\mathbb Z}_+}$ be the family of orthogonal projectors subordinated to the shells, $\displaystyle \Pp_{\Sb_n} = \sum_{j\in {\mathbb S}_n} | \delta_j \rangle \langle \delta_j | $.
 Similarly as in \eqref{eq-V} the unitary operator ${\mathcal V}$ is defined by
\begin{align*}
{\mathcal V} = \bigoplus_{n\in \ZZ_+}\, V_n,\quad \text{where}\quad  V_n = \Pp_{n} \Vv \Pp_{n},
\end{align*}
and $V_n$ is a
\begin{itemize}
\item clockwise shift for $n\in2\ZZ_+$  ($n$ even) where
$$
\begin{cases}
V_n\delta_{n,l}&= \delta_{n,l-1}\ , \qquad \quad\quad  l\in \{ -n, \ldots, n\}\\
V_n \delta_{l,-n-1} &= \delta_{l-1,-n-1}\ , \quad\quad \ l\in \{ -n, \ldots, n\} \\
V_n \delta_{-n-1,l} &= \delta_{-n-1,l+1}\ , \quad\quad\ l\in \{ -n-1, \ldots, n-1\} \\
V_n \delta_{l,n} &= \delta_{l+1,n}\ , \qquad \quad\quad  l\in \{ -n-1, \ldots, n-1\}.
\end{cases}
$$

\item counter-clockwise shift for $n\in2\ZZ_+ +1$ ($n$ odd) where

$$
\begin{cases}
V_n\delta_{n,l}&= \delta_{n,l+1}\ ,\;\;\quad \quad\quad l\in \{ -n-1, \ldots, n-1\} \\
V_n \delta_{l,n} &=  \delta_{l-1,n}\ , \quad \quad\quad \;\; l\in \{ -n, \ldots, n\} \\
V_n \delta_{-n-1,l} &= \delta_{-n-1,l-1}\ ,\quad\quad  l\in \{ -n, \ldots, n\} \\
V_n \delta_{l,-n-1} &=\delta_{l+1,-n-1}\ ,\quad\quad l\in \{ -n-1, \ldots, n-1\}. 
\end{cases}
$$
\end{itemize}

\vspace{.2cm}

See the example in the illustrations below in figure~\ref{fig5} for more details. We also could add some phases along the shifts.
To describe the operator $\Ww$, we pick two sequences $(a_n)_{n\in {\mathbb Z^*_+}}, (b_n)_{n\in {\mathbb Z_+^*}}$ in ${\mathbb Z}^2$ so that $a_n \in {\mathbb S}_{n-1},\ b_n \in {\mathbb S}_{n}$ with $a_{n+1} \neq b_n$ and $\| a_n - b_n\|_{\infty} =1$.
We connect the shells  through the vectors 
$e_{(n-1,+)} = \delta_{a_n}$ and $e_{(n,-)} = \delta_{b_n}$. 
Then we define the corresponding orthogonal projectors 
$$\Qq_n=| \delta_{a_{n}} \rangle \langle \delta_{a_{n}} | + | \delta_{b_n} \rangle \langle \delta_{b_n} |,\quad \Qq=\sum_{n\in\ZZ_+^*} \Qq_n,$$
and the unitary operator ${\mathcal W}$ is defined by
\begin{align*}
{\mathcal W} = \Qq^{\perp} \oplus \bigoplus_{n\in {\mathbb Z}_+^*} W_n,\quad W_n = \Qq_n {\mathcal W} \Qq_n,
\end{align*}
where $W_n$ are $2\times 2$ unitary  matrices.

\begin{figure}[h!]
\begin{center}
\includegraphics[scale=0.25]{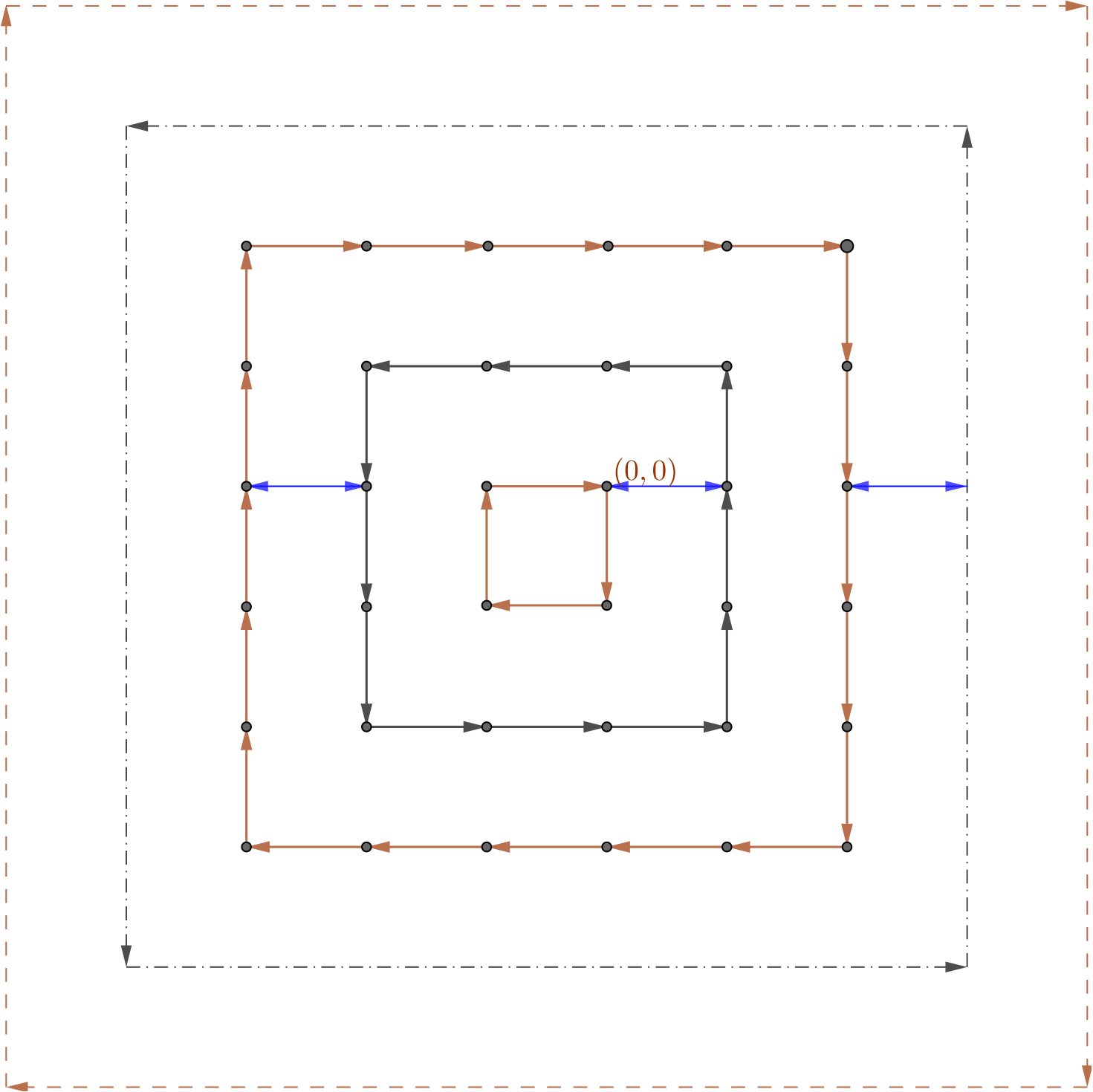}
\vspace*{-0.2cm}
\caption{The Stroboscopic model with a particular configuration where $a_{2n} = (-2n,0)\in \Sb_{2n-1},\ a_{2n+1}= (2n,0)\in \Sb_{2n},\ b_{2n} = (-2n-1,0)\in \Sb_{2n}$ for $n\in {\mathbb Z_+^*}$, \ and  $b_{2n+1} = (2n+1,0)\in \Sb_{2n+1}$ for $ n\in {\mathbb Z}_+ $.} \label{fig5}
\end{center}
\end{figure}

\section{Transfer matrix and Resolvent\label{sec-tr}}

In this section we want to relate the Green's function of the restrictions
$\Uu^{(u,v)}_N$, $\tilde \Uu^{(u,v)}_N$ to
the transfer matrix from $0$ to level $N$,  $T_{z,[0,N]}$.
Therefore, we introduce the notations
\begin{equation}
Q_{0,N}\,=\,\pmat{e_{(0,-)} & e_{(N,+)} }\,\in\,\CC^{\GG_N \times 2},\quad
P_{0,N}=\one_{\GG_N}-Q_{0,N} Q_{0,N}^*\,\in\,\CC^{\GG_N \times \GG_N},
\end{equation}
where we interpret $e_{(0,-)}$, $e_{(N,-)}$ as column vectors in $\CC^{\GG_N}$.
Then,  we define the boundary resolvent matrix from $0$ to $N$ by
\begin{equation}
R^{(u,v)}_{z,[0,N]}\,=\, Q_{0,N}^* (z^{-1}\Uu^{(u,v)}_N -\one_{\GG_N})^{-1} Q_{0,N}
\end{equation}
\begin{equation}
\tilde R^{(u,v)}_{z,[0,N]}\,=\, Q_{0,N}^* (z^{-1} \tilde \Uu^{(u,v)}_N -\one_{\GG_N})^{-1} Q_{0,N} \;.
\end{equation}
Note that this means
$$
R^{(u,v)}_{z,[0,N]}\,=\,
\pmat{e_{(0,-)}^*(z^{-1} \Uu^{(u,v)}_N -\one_{\GG_N})^{-1} e_{(0,-)}  & 
e_{(0,-)}^* (z^{-1}  \Uu^{(u,v)}_N -\one_{\GG_N})^{-1}  e_{(N,+)} \\ 
e_{(N,+)}^* (z^{-1}  \Uu^{(u,v)}_N -\one_{\GG_N})^{-1} e_{(0,-)} & 
e_{(N,+)}^* (z^{-1} \Uu^{(u,v)}_N -\one_{\GG_N})^{-1} e_{(N,+) }}
$$
and similar for $\tilde R$, replacing $\Uu$ with $\tilde \Uu$.
We obtain the following relations.

\begin{proposition}\label{prop-rel-T-R}
For any $N$ we find
\begin{align*}
T_{z,[0,N]}\,&=\,\varphi_\sharp\big(Q_{0,N}^*(z^{-1}\Uu^{(1,1)}_{N}-P_{0,N})^{-1} Q_{0,N}\big) \\
&=\, \pmat{1 & 1\\ 0 & 1} \varphi_\sharp(R^{(1,1)}_{z,[0,N]})  \pmat{1 & 0  \\ -1 & 1} \\
&=\,\pmat{1 & 1\\ 0 & 1} \varphi_\sharp(\tilde R^{(1,1)}_{z,[0,N]})  \pmat{1 & 0  \\ -1 & 1}.
\end{align*}
More generally  for any $N$,  $u,v \in \CC$ and $z$ where all quantities are well defined,  we have

\begin{align*}
\varphi_\sharp\big(R^{(u,v)}_{z,[0,N]}\big)
&=\,\pmat{1&-1 \\ 0 & 1} \pmat{v & 0 \\ 0 & 1} T_{z,[0,N]} \pmat{ 1 & 0 \\ 0 & u^{-1}} \pmat{1&0\\ 1 & 1}\\
\varphi_\sharp\big(\tilde R^{(u,v)}_{z,[0,N]}\big)
&=\,\pmat{1&-1 \\ 0 & 1}\pmat{1&0\\0&v^{-1}} T_{z,[0,N]} \pmat{u&0\\0&1}\pmat{1&0\\ 1 & 1}.
\end{align*}

\end{proposition}

\begin{rem}
From the last formulas above, we may focus on the operator $\Uu^{(u)}$ and include the boundary condition at the root by replacing
$T_{z,0}$ with $T_{z,0} \pmat{1&0\\0&u^{-1}}$,  while focusing on the operator $\tilde \Uu^{(u)}$ one might include the boundary condition by replacing
$T_{z,0}$ with $T_{z,0} \pmat{u&0\\0&1}$.  \\
The boundary condition $\Phi_{(-1,+)}=u\Psi_{(-1,+)}$ means to start with a multiple of the vector $\pmat{u\\1}$ in order to create  (formal) solutions to the eigenvalue equation.
In both cases,  including the boundary condition into the transfer matrices, this means to start with a multiple of the vector $\pmat{1\\1}$.
\end{rem}

\begin{proof}
Let us start with the case $N=0$,  and boundary conditions $u=v=1$ where $\Uu_0^{(1,1)}=V_0=\tilde \Uu_0^{(1,1)}$.
Then, we have 
$$T_{z,[0,0]}=T_{z,0}=\varphi_\sharp(Q_0^*(z^{-1}V_0-P_0)^{-1} Q_0),$$ 
and we want to relate this to the parts of the resolvent
$$
R=R^{(1,1)}_0=Q_0^*(z^{-1} V_0-\one)^{-1} Q_0\;.
$$
By the resolvent identity and using $\one-P=Q_0Q_0^*$ one obtains
$$
z^{-1} V_0-\one=(z^{-1}V_0-P_0)^{-1}+(z^{-1}V_0-\one)^{-1}Q_0Q_0^*(z^{-1} V_0-P_0)^{-1},
$$
which gives
\begin{equation}\label{eq-V-R}
Q_0^*(z^{-1}V_0-P_0)^{-1} Q_0\,=\,(\one+R)^{-1}\,R.
\end{equation}
Using Proposition~\ref{prop-A0.5}~b) this leads to
\begin{equation}\label{eq-T-R}
T_{z,0}=\varphi_\sharp((\one+R)^{-1} R)\,=\,\pmat{1&1\\0&1} \varphi_\sharp(R) \pmat{1&0\\-1&1}\;,
\end{equation}
which gives the first statement in the case $N=0$.
Now, with boundary conditions $(u, v)$ we have
$$
\Uu_0^{(u,v)}=\Ww^{(u,v)}_0 V_0\;, \qquad \tilde \Uu^{(u,v)}_0=V_0 \Ww^{(u,v)}_0,
$$
and a similar relation as in \eqref{eq-V-R} holds, using the resolvent matrices with boundary conditions $(u,v)$.
It is easy to see that
$$
\Ww^{(u,v)}_0 P_0=P_0 \Ww^{(u,v)}_0=P_0,\qquad \Ww^{(u,v)}_0 Q_0=Q_0 \pmat{u&0\\0&v}, \qquad\big( \Ww^{(u,v)}_0\big)^{-1}=\Ww^{(u^{-1},v^{-1})}_0\;.
$$
Therefore, one obtains
$$
Q_0^* (z^{-1}\Ww^{(u,v)}_0 V_0-P_0)^{-1}Q_0\,=\,Q_0^*(z^{-1}V_0-P_0)^{-1}Q_0 \pmat{u^{-1} & 0 \\ 0 & v^{-1}}
$$
$$
Q_0^* (z^{-1} V_0\Ww^{(u,v)}_0-P_0)^{-1}Q_0\,=\, \pmat{u^{-1} & 0 \\ 0 & v^{-1}} Q_0^*(z^{-1}V_0-P_0)^{-1}Q_0\;,
$$
and using Proposition~\ref{prop-A0.5} one obtains
$$
\varphi_\sharp\left( Q_0^* (z^{-1}\Ww^{(u,v)}_0 V_0-P_0)^{-1}Q_0\right)\,=\,\pmat{v&0 \\ 0& 1} T_{z,0} \pmat{1 & 0 \\ 0 & u^{-1}}
$$
and
$$
\varphi_\sharp\left( Q_0^* (z^{-1} V_0 \Ww^{(u,v)}_0-P_0)^{-1}Q_0\right)\,=\,\pmat{1&0 \\ 0& v^*} T_{z,0} \pmat{u & 0 \\ 0 & 1}.\;
$$
Then the relations \eqref{eq-T-R} hold replacing $R$ with $R^{(u,v)}_{z,[0,0]}$ (or $\tilde R^{(u,v)}_{z,[0,0]}$) and $T_{z,0}$ with the corresponding matrix as above,  meaning
$$
\pmat{v&0 \\ 0& 1} T_{z,0} \pmat{1 & 0 \\ 0 & u^{-1}}\,=\,\pmat{1&1\\0&1} \varphi_\sharp(R^{(u,v)}_{z,[0,0]}) \pmat{1&0\\-1&1}
$$
$$
\pmat{1&0 \\ 0& v^{-1}} T_{z,0} \pmat{u & 0 \\ 0 & 1}\,=\,\pmat{1&1\\0&1} \varphi_\sharp(\tilde R^{(u,v)}_{z,[0,0]}) \pmat{1&0\\-1&1},
$$
giving the second statement in the case $N=0$. 


Now, let us look at the effects of `grouping' the first $N+1$ shells $\Sb_0$, \ldots, $\Sb_N$ into a single shell $\GG_N$.  
Then,  using the splitting $\GG=\GG_N \sqcup \bigsqcup_{n=N+1}^\infty \Sb_n$ we find
$$
\Uu^{(u)}\,=\,\underbrace{\pmat{u \\ & \one_{|\GG_N|-2} \\ & & W_{N+1} \\ & & & \ddots}}_{=: \widehat \Ww^{(u)}} \underbrace{\pmat{\Uu_N^{(1,1)} \\ & V_{N+1} \\ & & \ddots}}_{=:\widehat \Vv}=\widehat \Ww^{(u)} \widehat \Vv\;
$$ 
and
$$
\tilde \Uu^{(u)}\,=\,\underbrace{\pmat{\tilde \Uu_N^{(1,1)} \\ & V_{N+1} \\ & & \ddots}}_{=: \breve \Vv} \pmat{u \\ & \one_{|\GG_N|-2} \\ & & W_{N+1} \\ & & & \ddots}\,=\,\breve \Vv\,\widehat \Ww^{(u)}\;,
$$
with $\Uu_N^{(1,1)},\; \tilde \Uu_N^{(1,1)}$ as in \eqref{eq-U-G_N}.
For the first transfer matrix in this setup, we consider the pairs of operators similar as before and introduce
$$
\widehat{\Uu}^{(u)}\,:=\,  \widehat \Vv \,\widehat \Ww^{(u)}\qtx{and} \breve \Uu^{(u)}\,=\,\widehat \Ww^{(u)}\breve \Vv.
$$
Now,  $(\Uu^{(u)}, \widehat\Uu^{(u)})$  is a corresponding pair of conjugated one-channel operators, as well as $(\breve \Uu^{(u)}, \tilde \Uu^{(u)})$\;.
Additional to the solutions $\Psi$ and $\Phi$ to the eigenvalue equations $\Uu^{(u)}\Psi=z \Psi$, $\tilde \Uu^{(u)}\Phi=z \Phi$ with the relation $\Ww^{(u)} \Phi=\Psi$ as in Proposition~\ref{prop-eig-equation} we define $\widehat \Phi$ and $\breve \Psi$ by
$$
\widehat \Ww^{(u)} \widehat \Phi\,=\,\Psi\;,\qquad \widehat \Ww^{(u)} \Phi\,=\,\breve \Psi.
$$
Then
\begin{equation}\label{eq-hat-Phi}
\widehat \Phi\,=\,\big(\widehat \Ww^{(u)}\big)^{-1} \Ww^{(u)} \Phi\,=\,\pmat{\Ww_N^{(1,1)} \\ & \one_{\GG\setminus \GG_N}  } \Phi
\end{equation}
and
\begin{equation}\label{eq-breve-Psi}
\breve \Psi\,=\, \widehat \Ww^{(u)} \big(\Ww^{(u)}\big)^{-1}\,\Psi\,=\,\pmat{\big(\Ww_N^{(1,1)}\big)^* \\ & \one_{\GG\setminus \GG_N}} \Psi.
\end{equation}
Note that in our conventions
$$
\pmat{\Phi_{(N,+)} \\ \Psi_{(N,+)}}\,=\,T_{z,[0,N]}\,\pmat{\Phi_{(-1,+)}\\ \Psi_{(-1,+) } }=T_{z,[0,N]} \pmat{\Psi_{(0,-)}\\ \Phi_{(0,-)}}.
$$
But using \eqref{eq-hat-Phi} and \eqref{eq-breve-Psi} we find
$$
\widehat \Phi_{(0,-)}=\Phi_{(0,-)}, \quad \widehat \Phi_{(N,+)}=\Phi_{(N,+)}, \quad
\breve \Psi_{(0,-)}=\Psi_{(0,-)},\quad  \breve \Psi_{(N,+)}=\Psi_{(N,+)}\;,
$$
and thus
$$
\pmat{\widehat \Phi_{(N,+)} \\ \Psi_{(N,+)}}=T_{z,[0,N]} \pmat{\Psi_{(0,-)} \\ \widehat \Phi_{(0,-)}} \qtx{and}
\pmat{\Phi_{(N,+)} \\ \breve \Psi_{(N,+)}}=T_{z,[0,N]} \pmat{\breve \Psi_{(0,-)} \\ \Phi_{(0,-)}}.
$$
We note that this is true for any boundary condition $u$ with $\Psi_{(0,-)}=u\Phi_{(0,-)}$.
Considering the operator pair $(\Uu^{(u)}, \widehat \Uu^{(u)})$ or $(\breve \Uu^{(u)}, \tilde \Uu^{(u)})$ and
comparing to
\eqref{eq-def-T-sharp} ,
\eqref{eq-T-sharp-1} we find
\begin{equation*}
T_{z,[0,N]}\,=\,\varphi_\sharp\big( Q_{0,N}^* (z^{-1} \Uu^{(1,1)}_N -P_{0,N})^{-1} Q_{0,N}\big)\,=\,
\varphi_\sharp\big(  Q_{0,N}^* (z^{-1} \tilde \Uu^{(1,1)}_N - P_{0,N})^{-1}  Q_{0,N}\big).
\end{equation*}
Thus, the relation between $T_{z,[0,N]}$ and $\Uu^{(1,1)}_N$ or $\tilde \Uu^{(1,1)}_N$ is the same as the relation of $T_{z,0}=T^\sharp_{z,0}$ with $V_0$ and the case $N=0$ above gives the general result. 
\end{proof}



\section{Proof of Theorem~\ref{th-main} and Theorem~\ref{th-main2} \label{sec-sp-av}}
In this section we will first prove Theorem~\ref{th-main} using spectral averaging techniques, and then obtain Theorem~\ref{th-main2} as a corollary.
Let $\mu^{(u,v)}_N$ be the spectral measure of $\Uu^{(u,v)}_N$ (or $\tilde \Uu^{(u,v)}$) at the vector $e_{(0,-)}$,  meaning,
$$
\mu^{(u,v)}_N(f)\,=\,\langle e_{(0,-)} \,, f(\Uu^{(u,v)}_N) e_{(0,-)} \rangle \;=\;
\langle e_{(0,-)} \,, f(\tilde \Uu^{(u,v)}_N) e_{(0,-)} \rangle,
$$
for Borel functions $f$ on the unit circle $\UU(1)$. Note, the second equality is easy to see using
$ f(\tilde \Uu^{(u,v)}_N)=(\Ww^{(u)}_N)^* f(\Uu^{(u,v)}_N) \Ww^{(u)}_N$ and $\Ww^{(u)}_N e_{(0,-)} = u e_{(0,-)}$ with $|u|=1$.
Next, we consider an average over the Haar measure for $v \in \UU(1)$ and define the measure $\mu^{(u)}_N$ by
\begin{equation}
\mu^{(u)}_N(f) = \frac{1}{2\pi} \int_{0}^{2\pi} \mu^{(u,e^{i\varphi})}_N (f)\,{\rm d\varphi}.
\end{equation}
We will denote the value of the Green's function by $g_N^{(u,v)}(z)$, meaning
$$
g^{(u,v)}_N(z):=e_{(0,-)}^*\big(z^{-1} \Uu_N^{(u,v)}-\one\big)^{-1} e_{(0,-)}=
\int \frac{1}{z^{-1}w-1} \,{\rm d}\mu^{(u,v)}_N(w),
$$
and we define the averaged value
$$
 g^{(u)}_N(z):= \int \frac{1}{z^{-1}w-1} \,{\rm d}\mu^{(u)}_N(w)\,=\,
\frac{1}{2\pi}\int_0^{2\pi} g^{(u,e^{i\varphi})}_N(z)\,{\rm d}\varphi\;.
$$

Next, we want to relate $\mu^{(u)}_N$ to the transfer matrix.

\begin{lemma}\label{lem-4.1}
Denote the entries of the transfer matrix by
$$
T_{z,[0,N]}\,=\,\pmat{A_z & B_z \\ C_z & D_z}.
$$
Then, for $|z|<1$, $z\not\in\widehat\Aa_N$, the averaged Green's function is given by
$$
g^{(u)}_N(z)\,=\,\frac{-B_z}{A_z u+B_z}\;.
$$
\end{lemma}
\begin{proof}
Using Proposition~\ref{prop-rel-T-R} one has
\begin{align*}
\varphi_\sharp(R_{z,[0,N]}^{(u,v)})
&=\pmat{1&-1\\0&1} \pmat{vA_z&vB_z u^{-1} \\ C_z & D_zu^{-1}} \pmat{1&0\\1&1}\\
&=\pmat{vA_z-C_z+vB_zu^{-1}-D_zu^{-1} & vB_zu^{-1} -D_zu^{-1}\\ C_z +D_zu^{-1} & D_z u^{-1}}.
\end{align*}
Note that if $R_{z,[0,N]}^{(u,v)}=\smat{\alpha & \beta \\ \gamma & \delta}$, then $g^{(u,v)}_N(z)=\alpha$ is the upper right entry and
 $\varphi_\sharp(R_{z,[0,N]}^{(u,v)})\,=\, \smat{\beta^{-1} & -\beta^{-1} \alpha \\ \delta \beta^{-1} & \gamma-\delta \beta^{-1}\alpha}$. Thus, using
 $\alpha=\beta^{-1}\alpha / \beta^{-1}$ we obtain
$$
g^{(u,v)}_N(z)=\frac{-(vB_zu^{-1}-D_zu^{-1})}{vA_z-C_z+vB_z u^{-1}-D_z u^{-1}}=\frac{-(B_z-v^{-1}D_z)}{(A_z-v^{-1}C_z)u+B_z-v^{-1}D_z}.
$$
Extending the definitions to $v \in \CC$, we have
$$
g^{(u,v)}_N\,=\,e_{(0,-)}^*\left(z^{-1} \pmat{u \\ & \one \\ & & v}\Uu^{(1,1)}_N-\one \right)^{-1} e_{(0,-)}.
$$
Now for $|u|=1$ and $|z|<1$ fixed, the function is holomorphic in $v$ for $|v|>|z|$, or holomorphic in $v^{-1}$ for $|v^{-1}|<|z|^{-1}$ where $|z|^{-1}>1$.
Thus, the average over $v\in\UU(1)$ simply means to replace $v^{-1}$ by $0$, giving the desired formula.
\end{proof}

\begin{lemma}\label{lem-mu_N}
There is a point measure $\nu_N^{(u)}$ supported on the finite set $\Aa_N$, such that for $f \in C(\UU(1))$
$$
\mu^{(u)}_N(f)\,=\, \nu^{(u)}_N(f)\,+\, \int_0^{2 \pi} \frac{f(e^{i\varphi})}{\pi \| T_{e^{i\varphi},[0,N]}\smat{u\\1}\|^2}\,{\rm d}\varphi\,.
$$
\end{lemma}
\begin{proof}
Consider the Poisson transform of the measure $\mu^{(u)}_N$, that is, using Lemma~\ref{lem-4.1}
\begin{align*}
P^{(u)}(z)
&=\Re e\int \frac{z^{-1}w+1}{z^{-1}w-1}\, {\rm d}\mu^{(u)}_N(w)\,=\,1+2\,\Re e \,g^{(u)}_N(z)\\
&=\frac{|A_z u+B_z|^2-B_z \overline{(A_zu+B_z})-\overline{B_z}(A_zu+B_z)}{|A_z u+B_z|^2}\,=\,\frac{|A_z|^2-|B_z|^2}{|A_z u+B_z|^2}\;
\end{align*}
for $|z|<1$, $z\not\in\widehat\Aa_N$.
Note, for $z=e^{i\varphi}\not \in \Aa_N=\widehat\Aa_N\cap \UU(1)$ we have $T_{z,[0,N]}\in \UU(1,1)$ which implies $T^*_{z,[0,N]}\in\UU(1,1)$ and as such $|A_z|^2-|B_z|^2=1$.
Moreover,
\begin{align*}
\left(T_{z,[0,N]} \pmat{u\\1}\right)^*\pmat{1\\ & -1} T_{z,[0,N]} \pmat{u\\1} 
&=\pmat{u^*&1} T^*_{z,[0,N]} \pmat{1\\&-1}T_{z,[0,N]} \pmat{u\\1}\\
&=\pmat{u^* & 1} \pmat{1\\&-1}\pmat{u\\1}=|u|^2-1=0\;,
\end{align*}
implying $|A_zu+B_z|^2=|C_z u+D_z|^2$ and hence $2\,|A_z u +B_z|^2=\|T_{z,[0,N]}\smat{u\\1}\|^2$.\\
Thus, for $e^{i\varphi} \not \in \Aa_N$ 
\begin{equation}\label{eq-lim-r-1}
\lim_{r \nearrow 1}\,P(re^{i\varphi})\,=\,\frac{1}{\left|A_{e^{i\varphi}} u+B_{e^{i\varphi}}\right|^2}=\frac{1}{\left| \pmat{1&0} T_{e^{i\varphi},[0,N]}\smat{u\\1}\right|^2}=\frac{2}{\| T_{e^{i\varphi},[0,N]} \smat{u\\1}\|^2}\;.
\end{equation}
Let us note that ${\rm d}\mu^{(u)}_N(e^{i \varphi})$ is the weak limit of
$\frac{1}{2\pi} P^{(u)}(re^{i\varphi}) \,{\rm d}\varphi$ for $r \nearrow 1$.
Hence, the measure $\mu^{(u)}_N$ is absolutely continuous in $\UU(1)\setminus \Aa_N$ with respect to the normalized Haar measure on $\UU(1)$ and the density is given by the right hand side of \eqref{eq-lim-r-1} divided by $2\pi$.
Moreover, as $\Aa_N$ is a finite set,  the restriction of $\mu^{(u)}_N$ to $\Aa_N$ is a point measure $\nu_N^{(u)}$.
This finishes the proof. 
\end{proof}

\begin{lemma} \label{lem-limit-nu} We have 
$\nu^{(u)}_{N+1} \geq \nu^{(u)}_N$ and $\displaystyle\lim_{N\to\infty} \nu_N^{(u)}=\nu^{(u)}$ with $\nu^{(u)}$ as defined in \eqref{eq-def-nu^u}.
\end{lemma}
\begin{proof}
If for $a\in \Aa_N$ we have $\nu_N^{(u)}(\{a\})=\mu^{(u)}_N(\{a\})>0$, then this means that $a$ is an eigenvalue of $\Uu^{(u,v)}_N$ for a set of positive measure in $v\in\UU(1)$.
By rank one perturbation arguments, at the vector $e_{(N,+)}$, the spectral average over $v =e^{i\varphi} \in \UU(1)$ with respect to the Haar measure gives the Haar measure on $\UU(1)$,  see for instance \cite[Proposition~8.1]{Bou}\footnote{where in  the notations of \cite{Bou} we use Proposition~8.1 with $T=1$, $e^{-iH_0}=\Uu_N^{(u,1)}$, $\phi=e_{(N,+)}$ and the average over $\kappa \in \TT=\RR/2\pi\ZZ$ corresponds to the average over $\varphi$}. This means, for $f \in C(\UU(1))$
\begin{align} \nonumber
&\int_0^{2\pi} \,\left\langle e_{(N,+)},f\left(\Uu^{(u,e^{i\varphi})}_N\right) e_{(N,+)}\,\right\rangle\frac{\rm d\varphi}{2\pi}\\
&\quad=\int_0^{2\pi} \,\left\langle e_{(N,+)},f\left(e^{i\varphi |e_{(N,+)}\rangle\langle e_{(N,+)} |}\,\Uu^{(u,1)}_N \right) e_{(N,+)}\,\right\rangle\frac{\rm d\varphi}{2\pi}
=\int_0^{2\pi}f(e^{i\theta})\frac{\rm d\theta}{2\pi}\,. \label{eq-sp-av-Haar}
\end{align}
Now fix $a\in \UU(1)$ and consider the restriction of $\Uu^{(u,v)}_N$ to the cyclic space $\Zz_N$ generated by $e_{(N,+)}$. Note that $\Zz_N$ does not depend on $v$. 
By \eqref{eq-sp-av-Haar}, the set
$$\{v\in\UU(1)\,:\, a\,\text{is eigenvalue of the restriction } \Uu^{(u,v)}_N|_{\Zz_N}\,\}$$ 
has Haar measure zero.
Therefore, if $\nu_N^{(u)}(\{a\})>0$, then, for some $v$ there is an eigenvector $\Psi \in \Zz_N^\perp$, meaning $\Psi^* e_{(N,+)}=0$. Such an eigenvector $\Psi$ is in fact an eigenvector of $\Uu^{(u,v)}_N$ for all $v \in \UU(1)$, with the same eigenvalue $a$. Thus, we find that $\nu^{(u)}_N(\{a\})$
is precisely given by the norm squared of $e_{(0,-)}$ projected to $\ker(\Uu^{(u,v)}_N -a\one)\cap \Zz_N^\perp$.

Furthermore, if $\Psi\in \ker(\Uu^{(u,v)}_N -a\one)\cap \Zz_N^\perp$, then, for $M>N$ we find that $\Psi\oplus \nul=\hat \Psi\in \ker(\Uu^{(u,v)}_M -a\one) \cap \Zz_M^\perp$, implying, $\nu_M(\{a\}) \geq \nu_N(\{a\})$. 
Together with the uniform boundedness of the probability measures $\mu^{(u)}_N$ this means that there is a limit point measure
$$
\tilde \nu^{(u)}\,=\,\lim_{N \to \infty} \nu_N^{(u)}\;,
$$
which is supported on the set $\Aa$.
Further note that, in fact, $\Psi \oplus \nul\in \ell^2(\GG)=\ell^2(\GG_N) \oplus \ell^2(\GG\setminus \GG_N)$ is also an eigenfunction of $\Uu^{(u)}$ with eigenvalue $a$ in this case. Any finitely supported eigenfunction of $\Uu^{(u)}$ is of this form and induces an eigenfunction of $\Uu^{(u,v)}_N$ in $\Zz_N^\perp$ for $N$ large enough. Thus, $\tilde \nu^{(u)}$ is precisely the part of the spectral measure $\mu^{(u)}$ coming from the space $\HH^{(u)}_c$ generated by finitely supported eigenfunctions and $\tilde \nu^{(u)}=\nu^{(u)}$ as defined in \eqref{eq-def-nu^u}.
\end{proof}

Now we are ready to prove Theorem~\ref{th-main}.
\begin{proof}[Proof of Theorem~\ref{th-main}]
Using resolvent convergence we have 
$$g^{(u,v_N)}_N(z) \to g^{(u)}(z):=e_{(0,-)}^*\big(z^{-1} \Uu^{(u)}-\one\big)^{-1} e_{(0,-)}\;=\;\int \frac{1}{z^{-1}w-1} \,{\rm d}\mu^{(u)}(w)\;.$$ 
 for any $|z|<1$ and any sequence $v_N \in \UU(1)$ as $N \to \infty$. 
In particular, this means that the compact sets $\{ g^{(u,v)}_N(z)\,:\,v\in\UU(1)\}$ shrink to a point and we also find convergence for the averages, $g^{(u)}_N(z) \to g^{(u)}(z)$.
This implies the convergence of the Poisson transforms and hence $\mu^{(u)}_N$ converges weakly to $\mu^{(u)}$. Thus, with Lemma~\ref{lem-mu_N} and Lemma~\ref{lem-limit-nu} we find
$$
{\rm d}\mu^{(u)}(e^{i \varphi})\,=\,{\rm d}\nu^{(u)}(e^{i\varphi})\,+\,\lim_{n \to \infty}
\frac{1}{\pi} \,\frac{{\rm d}\varphi}{\left\|T_{e^{i\varphi},[0,N]} \smat{u\\1} \right\|^2}\;
$$
weakly,
proving Theorem~\ref{th-main}.
\end{proof}

\begin{proof}[Proof of Theorem~\ref{th-main2}]
The proof works exactly the same as in \cite{Lasi, Sa-AT, Sa-OC, Sa-Tr, Si-OPUC}.
As for $\varphi \in \RR$ we find $T_{e^{i\varphi},[0,n]}\in \UU(1,1)$ , one has
$$
\pmat{0 & -1}T_{e^{i\varphi},[0,n]}^* \pmat{1 \\ & -1}  T_{e^{i\varphi},[0,n]} \pmat{u\\1}\,=\,\pmat{0 & -1}  \pmat{1 \\ & -1} \pmat{u\\1}\,=\,1\;.
$$
Then, the Cauchy-Schwartz inequality gives
$$
1\,\leq\, \left\| \pmat{1 \\ & -1} T_{e^{i\varphi},[0,n]} \pmat{0 \\ -1}   \right\| \; \left\| T_{e^{i\varphi},[0,n]} \pmat{u\\1}  \right\| \, \leq\, 
\| T_{e^{i\varphi},[0,n]}\|\, \left\| T_{e^{i\varphi},[0,n]} \pmat{u\\1}  \right\|,
$$
and hence 
$$
\frac{1}{\left\|T_{e^{i\varphi},[0,n]} \smat{u\\1} \right\|^2}\,\leq\, \| T_{e^{i\varphi},[0,n]}\|^2\,.
$$
Thus, a uniform bound on
$$
\int_{\varphi_0}^{\varphi_1} \| T_{e^{i\varphi},[0,n]}\|^{2p}\, {\rm d}\varphi\,<\,C,
$$
means that  $\varphi\mapsto \frac{1}{\left\|T_{e^{i\varphi},[0,n]} \smat{u\\1} \right\|^2}$ is a bounded sequence in $L^p(\varphi_0,\varphi_1)$ and has a weakly convergent subsequence with a limit function $\varphi\mapsto g(e^{i\varphi})$ in $L^p$ (note $p>1$).
Then, Theorem~\ref{th-main} assures that for $f$ continuous on $\UU(1)$ and compactly supported on 
$\{ e^{i\varphi}\,:\,\varphi\in [a,b]\}$ with $\varphi_0<a<b<\varphi_1$, $b-a \leq 2\pi$ we find
$$
\mu^{(u)}(f)-\nu^{(u)}(f)\,=\, \int_{\varphi_0}^{\varphi_1} f(e^{i \varphi})\, g(e^{i\varphi})\, {\rm d}\varphi,
$$
meaning that $\mu^{(u)}-\nu^{(u)}$ is an absolutely continuous measure w.r.t. ${\rm d}\varphi$ on $e^{i(\varphi_0,\varphi_1)}$ and has an $L^p$ density.
This shows the first part  of Theorem~\ref{th-main2}. 
For the second part, let $[a,b]\subset (\varphi_0,\varphi_1)$, $f\geq 0$  and $f(e^{i\varphi})>\varepsilon>0$ for all $\varphi\in[a,b]$.
Using Jensen's inequality for the convex function $F(x)=x^{-1/p}$ and the `random variable' $X(z)=\|T_{z,[0,n]}\smat{u\\1}\|^{2p}$ we find
$$
\frac{1}{b-a} \int_a^b \frac{{\rm d}\varphi}{\|T_{e^{i\varphi},[0,n]} \smat{u\\1}\|^2}\,\geq\,
\left( \frac{1}{b-a} \int_a^b \|T_{e^{i\varphi},[0,n]}  \smat{u\\1}\|^{2p}\,{\rm d}\varphi \right)^{-1/p}\,\geq\,\sqrt[p]{\frac{b-a}{C}}\;.
$$
Thus,
$$
(\mu^{(u)}-\nu^{(u)})(f)\,\geq\, \lim_{n\to\infty} \int_a^b  \frac{f(e^{i\varphi}){\rm d}\varphi}{\pi\,\|T_{e^{i\varphi},[0,n]} \smat{u\\1}\|^2}\geq\frac{\varepsilon (b-a)}
{\pi}\sqrt[p]{\frac{b-a}{C}}>0\,.
$$
\end{proof}

\section{Perturbation of periodic scattering zipper\label{sec-perper}}

In this section we will show Proposition~\ref{prop-sigma-ac} and Theorem~\ref{th-main3}.
We consider a periodic one-channel scattering zipper $\Uu=\Ww\Vv$ with period $p$ and the random perturbation $\Uu_\omega$ as defined
by  \eqref{eq-Uu_omega}. As one may put the boundary condition into re-defining $V_0$, we may assume that $u=1$.
We start with the following observation.
\begin{lemma}\label{lem-1}
Let $\Uu$ be a one-channel scattering zipper fulfilling assumptions {\rm (A1), (A2)}, meaning that $b_n\neq 0, \bbb_n\neq 0$ for all $n$. 
Then, the vector $e_{(0,-)}$ is cyclic, meaning 
$$
\overline{{\rm span}\{ \Uu^m e_{(0,-)}\,:\, m\in\ZZ\}}\,=\,\ell^2(\GG).
$$
\end{lemma}
Note that the Lemma also applies to the operators $\Uu_\omega$, almost surely.
\begin{proof}
Note that for scattering zippers, the vectors $e_{(n,\pm)}$ form an orthonormal basis of $\ell^2(\GG)$.
By induction we will prove that for all $n\in \ZZ_+$, $e_{(n,-)}, e_{(n,+)}$ are in the $\Uu$-cyclic space $\Zz$ generated by $e_{(0,-)}$.\\
For $n=0$ note $e_{(0,-)}\in\Zz$ and $\Uu^{-1} e_{(0,-)}\,=\,\Vv^* \Ww^* e_{(0,-)}=\Vv^* e_{(0,-)}=\bar a_0 e_{(0,-)}+ \bar b_0 e_{(0,+)}$, thus also $e_{(0,+)}\in\Zz$.\\
Now assume $e_{(m,-)}, e_{(m,+)}\in\Zz$ for all $m\leq n$. Then
$$
\Uu e_{(n,-)}\,=\, \Ww\left( a_n e_{(n,-)}+ c_n e_{(n,+)}\right) \,=\,\underbrace{a_{n} \Ww e_{(n,-)}\,+\,c_n \aaa_{n+1} e_{(n,+)}}_{\in \Zz}\,+\,c_n \bbb_{n+1} e_{(n+1,-)}.
$$
Note, $c_n \neq 0$ (as $b_n\neq 0$) and $\bbb_{n+1}\neq 0$ by assumption.
Hence, we see that $e_{(n+1,-)}\in \Zz$. Now, with the induction hypothesis, $\Ww e_{(n+1,-)}\,=\, \bbb_{n+1} e_{(n,+)}+\ddd_{n+1}  e_{(n+1,-)}\in\Zz$,
and therefore
$$
\Uu^{-1} \Ww e_{(n+1,-)}\,=\, \Vv^*  e_{(n+1,-)} \,=\,\bar a_{n+1} e_{(n+1,-)}\,+\, \bar b_{n+1} e_{(n+1,+)},
$$
and we see that $e_{(n+1,+)}\in\Zz$. This finishes the induction and the proof.
\end{proof}

For this reason, to analyze the spectrum of $\Uu, \Uu_\omega$ and its spectral types it is sufficient to consider the spectral measure at $e_{(0,-)}$ which we denote by
$$
\mu(f)\,=\,\langle e_{(0,-)},\, f(\Uu) e_{(0,-)}\rangle\;,\quad \mu_\omega(f)\,=\,\langle e_{(0,-)},\, f(\Uu_\omega) e_{(0,-)}\rangle\,.
$$

Furthermore, note that for $\Uu$ and $\Uu_\omega$ we have that the set $\Aa$ as defined in
\eqref{eq-Aa_n}, \eqref{eq-Aa} is empty. Hence, the measure $\nu$ (or $\nu^{(u)}$) as in Theorem~\ref{th-main2} is equal to zero.
As in Section~\ref{sub-random} we have the transfer matrix over a period
$T_{z}=T_{z,[1,p]}$ and recall $\Sigma=\{z\in\UU(1)\,:\,|\Tr T_z|<2 \}$.
We now show that $\Sigma$ is always a subset of the absolutely continuous spectrum (almost surely).
This is the main part of the proof of Theorem~\ref{th-main3} and Proposition~\ref{prop-sigma-ac}.
\begin{lemma}\label{lem-2}
There is a set $\Omega'\subset \Omega$ of probability one, such that for all $\omega \in \Omega'$, 
the spectral measure $\mu_\omega$ of $\Uu_\omega$ at the vector $e_{(0,-)}$ is purely absolutely continuous in $\Sigma$ and $\Sigma \subset \supp \mu_\omega$. We also find that $\mu$ is purely absolutely continuous in $\Sigma$ and $\Sigma \subset \supp \mu$. 
\end{lemma}
\begin{proof}
It is sufficient to consider $\Uu_\omega$ as it includes the case of the operator $\Uu$ when defining $\widehat V_n(\omega)=V_n, \widehat W_n(\omega)=W_n$ deterministically.

By Proposition~\ref{prop-U11}, for $z\in \Sigma$, there is  $\varphi_z, \theta_z \in \RR$, (not necesarily equal) and $M_z \in \GL(2,\CC)$ such that
\begin{equation}\label{eq-MTM}
M_z^{-1} T_z M_z\,=\, \pmat{e^{i\varphi_z} \\ & e^{i\theta_z}}=: R_z\,.
\end{equation}
Using that $T_z$ is analytic away from $0$, we can choose $\varphi_z, \theta_z$ and $M_z$ to depend analytically on $z$, locally on compact neighborhoods $e^{i[a,b]}\subset \Sigma$ of any $z\in \Sigma \subset \UU(1)$ (neighborhood within the space $\UU(1)$).
Within $e^{i[a,b]}$, $\|M_z\|$, $\|M_z^{-1}\|$ and $\|T_{z,0}\|$ are uniformly bounded by some constant $C$.

As above, the transfer matrix for the perturbed operator are indicated with a hat.
For convenience, we also define the products
$$
\widetilde T_{z,n}\,:=\,\widehat T_{z,[(n-1)p+1, np]} \qtx{for $n\geq 1$, and} \widetilde T_{z,0}= \widehat T_{z,0}^\sharp\;.
$$
This way
$$
\widehat T_{z,[0,np]}\,=\, \widetilde T_{z,n} \widetilde T_{z,n-1} \,\cdots \, \widetilde T_{z,1} \widetilde T_{z,0}\;,
$$
which is a perturbation of $T_z^n T_{z,0}$.
Now, as $V_n, \widehat V_n, W_n, \widehat W_n$ are unitary matrices, their norms are all equal to one. Thus, with assumption (C3), the matrices $\|T_{z,n}\|$,
$\|\widehat T_{z,n}\|$ are uniformly bounded by some constant $C$. Possibly increasing the constant $C$ above, we may use the same bound.

Using assumption (C1) we see that all $\widetilde T_{z,n}$ are independent.
The boundedness of $V_n \widehat V_n, W_n, \widehat W_n$ and $b_n^{-1}, \hat b_n^{-1}, \bbb_n^{-1}, \hat \bbb_n^{-1}$ and assumptions (C2), (C3) give
\begin{equation}\label{eq-sum-T}
\sum_{n=0}^\infty \left( \big\| \EE(\widehat T_{z,n})-T_{z,n} \big\|\,+\, \EE\big( \big\|\widehat T_{z,n}-T_{z,n}\big\|^2\big)\right)\,<\,\infty\;,
\end{equation}
uniformly in $z \in \UU(1)$.
We have
\begin{align}
&\widetilde T_{z,n}-T_z
\,=\, \prod_{k=1}^p \widehat T_{z,(n-1)p+k} \,-\, \prod_{k=1}^p  T_{z,(n-1)p+k} \nonumber \\ \label{eq-des-tildeT}
&=\sum_{k=1}^p\left( \left( \prod_{l=k+1}^p  \widehat T_{z,(n-1)p+l} \right) \,\left( \widehat T_{z,(n-1)p+k}-T_{z,(n-1)p+k}\right) \left( \prod_{l=1}^{k-1} T_{z,(n-1)p+l} \right) \right),
\end{align}
where the products go from right to left.
Thus,
\begin{equation}\label{eq-tildeT-T}
\| \widetilde T_{z,n} - T_{z} \|\,\leq\, C^{p-1} \sum_{k=1}^p \|\widehat T_{z,(n-1)p+k} - T_{z,(n-1)p+k}\|.
\end{equation}
Replacing $\widetilde{T}_{z,n}$ and $\widehat{T}_{z,n}$ with their expectations in \eqref{eq-des-tildeT} and using independence leads to
\begin{equation}\label{eq-EtildeT-T}
\big\| \EE\big( \widetilde T_{z,n}\big)  - T_{z} \big \|\,\leq\, C^{p-1} \sum_{k=1}^p \big \| \EE\big(\widehat T_{z,(n-1)p+k} \big)- T_{z,(n-1)p+k} \big \|\,
\quad\text{for $n\geq 1$}\,.
\end{equation}
Now, let us define the random, independent matrices
$$
S_{z,n}\,:=\,M_z^{-1} \widetilde T_{z,n} M_z - R_z.
$$
Then, using \eqref{eq-MTM}, \eqref{eq-sum-T}, \eqref{eq-tildeT-T}, \eqref{eq-EtildeT-T} and the fact that $\widetilde T_{z,0}$ is uniformly bounded, we obtain
\begin{equation}\label{eq-sum-S}
\sum_{n=0}^\infty \left( \big\| \EE( S_{z,n}) \big\|\,+\, \EE\big( \|  S_{z,n} \|^2 \big)\right)
\,<\, \tilde C < \infty,
\end{equation}
for some $\tilde C>0$ and all $z \in e^{i[a,b]} \subset \Sigma$.
Noting that additionally 
$$
\|S_{z,n}\|\,=\, \big\|  M_z^{-1} [\widetilde T_{z,n} - T_{z}] M_z \big\| \,\leq\, 2C^3,$$
for all $z \in E^{i[a,b]}$, we can copy the proof in \cite{KiLS}. Start with some fixed vector $\vv=\vv_{0}$ and consider the $z$ dependent random Markov process 
$$
\vv_{n+1}=M_z^{-1} \widehat T_{z,[0,np]} M_z \, \vv_0\,=\, (R_z+S_{z,n}) \vv_n.
$$
Then, $\vv_n$ is independent of $S_{z,n}$ and we find
\begin{align*}
\EE (\|\vv_{n+1}\|^4)\,&=\, \EE \left( \big( \|R_z \vv_n\|^2 + \|S_{z,n} \vv_n\|^2 +{\,\vv_n}^*( R_z^*S_{z,n}+S_{z,n}^* R_z)\vv_n \big)^2\right)\\
\,&=\, \EE \Big(\| \vv_n\|^4+ \big(  \|S_{z,n} \vv_n\|^2 +2\Re e ({\,\vv_n}^* R_z^*S_{z,n} \vv_n) \big)^2\,\Big)\, +\,\\
&\quad+\,
\EE \Big(2\|\vv_n\|^2\,\|S_{z,n} \vv_n\|^2) + 4 \, \Re e  \big({\,\vv_n}^* R_z^*\EE(S_{z,n}) \vv_n \big)\, \|\vv_n\|^2\,
\Big)\\
&\leq\, \EE(\|\vv_n\|^4)\left( 1+ (2C^3+2)^2 \,\EE\big( \|S_{z,n}\|^2 \big)+2 \EE\big( \|S_{z,n}\|^2 \big) + 4\| \EE(S_{z,n}) \| \right) \\
&\leq\, \EE(\|\vv_n\|^4) \exp\Big( c \; \big( \big( \| \EE(S_{z,n})\|\,+\,\EE(\|S_{z,n}\|^2) \Big),
\end{align*}
for some adequate $c>0$.
Iterating and using \eqref{eq-sum-S} this means
$$
\sup_n \EE (\|\vv_n\|^4)\,\leq\, \exp(c\,\tilde C)\,\|\vv_0\|\,.
$$
Using different $\vv_0$ which form an orthonormal basis of $\CC^2$, we conclude
$$
\sup_n\,\EE \left( \big\| M_z^{-1} \widehat T_{z,[0,np]} M_z \big\|^4 \right)\,\leq \,2^4  \exp(c\,\tilde C),
$$
for all $z \in e^{i[a,b]}$.
Thus,
$$
\EE\left( \liminf_{n\to\infty} \int_a^b \| \widehat T_{e^{i\varphi},[0,np]} \|^4\, {\rm d}\varphi\right)\,\leq\,
\liminf_{n \to \infty} \int_a^b \EE\big( \| \widehat T_{e^{i\varphi},[0,np]} \|^4\big) \,{\rm d}\varphi\,\leq\,
(b-a) 2^4C^8 e^{c\, \tilde C}\,.
$$
This means, that for a set of probability one, $\Omega_{a,b}$, and all $\omega \in \Omega_{a,b}$ we have
$$
\liminf_{n\to\infty} \int_a^b \| \widehat T_{e^{i\varphi},[0,np]} \|^4\,{\rm d}\varphi\,<\, \infty\,.
$$
By Theorem~\ref{th-main2}, for $\omega \in \Omega_{a,b}$, 
the measure $\mu_\omega=\mu_\omega-\nu_\omega$ is purely absolutely continuous in $e^{i(a,b)}$ and has support in all of $e^{i(a,b)}$.\\
Using such a neighborhood for any $z \in \Sigma$ with rational $a,b$, we obtain the open set $\Sigma$ as a countable union of such intervals on the unit circle. Then, the intersection $\Omega'$ of such $\Omega_{a,b}$ has probability one, and for all $\omega \in \Omega'$, the measure $\mu_\omega$ is purely absolutely continuous in $\Sigma$, and $\Sigma \subset \supp \mu_\omega$.
\end{proof}

Now, we are missing to see that $\overline{\Sigma}$ is in fact all of the essential spectrum of $\Uu$ and that it is not empty.
To obtain this we first note the following.
\begin{lemma}\label{lem-3}
The set $\{z \in \UU(1)\,:\,|\Tr T_z|=2\}$ is finite.
\end{lemma}
\begin{proof}
We note that by Proposition~\ref{prop-U11} we find for $z\in\UU(1)$ that $|\Tr T_z|=2 \Leftrightarrow \tfrac{(\Tr T_z)^2}{\det T_z} = 4$.
Moreover,  $z\mapsto\tfrac{(\Tr T_z)^2}{\det T_z} $ is analytic for $z\neq 0$. Therefore, if this would be equal to 4 on an infinite set within the unit circle $\UU(1)$, then the  function would have to be identically $4$, and so would be  $\frac{(\Tr z^p T_z)^2}{ \det(z^p T_z)}$. Now, using \eqref{eq-T-sharp-SZ}, \eqref{eq-T-flat-SZ}  one finds for $n\geq 1$
\begin{equation}\label{eq-lim-z0}
\lim_{z \to 0} z T_{z,n} = \pmat{e^{i\varphi_n} \bbb_n^{-1} b_n^{-1} & e^{i\chi_n} \bbb_n^{-1} d_n b_n^{-1}  \\ 0 & 0},  
\end{equation}
where $e^{i\varphi_n}=\det(V_nW_n), e^{i\chi_n}=-\det(W_n)$.
One therefore finds 
$$
\lim_{z\to 0} \det(z^p T_z)\,=\, 0 \qtx{and} \lim_{z \to 0} \Tr (z^p T_z)\,=\, \prod_{n=1}^p e^{i\varphi_n} \bbb_n^{-1} b_n^{-1} \,\neq\, 0,
$$
and thus $\displaystyle\lim_{z \to 0} \tfrac{(\Tr T_z)^2}{\det T_z} =\infty$ on the Riemann sphere. Hence,  $\tfrac{(\Tr T_z)^2}{\det T_z} $ is not constant.
Thus, the set is finite.
\end{proof}

Now let us classify the discrete spectrum of $\Uu$.
\begin{lemma}\label{lem-4}
The set of $z \in \UU(1)$ where $z$ is an eigenvalue of $\Uu$ is finite and coincides with the set where
$T_{z,0} \smat{1\\1}$ is an eigenvector of $T_z$ with eigenvalue $|\lambda_z|<1$. In particular, one has $|\Tr T_z|>2$ for such $z$.
\end{lemma}
\begin{proof}
The Lemma has two claims:\\
Claim 1: $z$ eigenvalue of $\Uu$ $\;\Leftrightarrow\;$ $T_{z,0} \smat{1\\1}$ is an eigenvector of $T_z$ with eigenvalue $|\lambda_z|<1$.\\
Claim 2: The set of eigenvalues is finite.\\[.2cm]
Let us first show Claim 1:  For the unique formal solutions (uniqueness follows because it needs to have the correct boundary condition) of $\Uu \Psi=z\Psi$ and $\Ww\Phi=\Psi$ we have
that
$$
\pmat{ \Phi_{(n,+)} \\ \Psi_{(n,+)}}\,=\, T_{z,[0,n]} \pmat{1\\1}\,, \quad
\pmat{\Phi_{(n,+)} \\ \Psi_{(n,+)} }\,=\,T_{z,n}^\sharp \,\pmat{\Psi_{(n,-)} \\ \Phi_{(n,-}) }.
$$
We have the following equivalences:
$$z\ \text{is an eigenvalue} \Leftrightarrow\;\Psi \in \ell^2(\GG)\;\Leftrightarrow\;\Phi\in\ell^2(\GG).$$
 By  Proposition~\ref{prop-U11}, if for $z\in\UU(1)$ the vector
 $$
\pmat{ \Phi_{(np,+)} \\ \Psi_{(np,+)}}\,=\, T_{z,[0,pn]} \pmat{1\\1}\,=\,T_z^n T_{z,0}\pmat{1\\1},
$$
is decaying, then $T_{z,0}\smat{1\\1}$ is an eigenvector of $T_z$ with eigenvalue $|\lambda_z|<1$.
This shows the direction `$\Rightarrow$' of Claim 1.\\
Now assume that $T_{z,0}\smat{1\\1}$ is an eigenvector of $T_z$ with eigenvalue $|\lambda_z|<1$.
Then, $|\Psi_{(np,+)}|, |\Phi_{{np,+}}|$ are exponentially decaying in the sense that they are bounded by 
$C|\lambda_z^{n}|$. Using the uniform boundedness of $T^\sharp_{z,n}$ and $T^\flat_n$ we get for $k=0,\ldots,p-1$ bounds of the form
$$
|\Psi_{(np+k,\pm)}| \leq  C \left\|\pmat{\Phi_{(np,+)} \\ \Psi_{(np,+)}} \right\|,
$$
with a uniform constant $C$ (uniform in $n$). Thus, $\Psi = \smat{\Psi_{(n,+)} \\ \Psi_{(n,-)}}_{n\in\ZZ_+} \in \ell^2(\GG)$ is an eigenvector and $z$ is an eigenvalue. This finishes the proof of Claim 1. $\hfill \blacksquare$

\vspace{.2cm}

\noindent For Claim 2, assume that there is an infinite number of eigenvalues.
Note that $T_{z,0}\smat{1\\1}$ is an eigenvector of $T_z$, iff 
$$\det\big( T_z T_{z,0}\smat{1\\1}, T_{z,0}\smat{1\\1}\big)\,=\,0\,.$$ The left hand side is a meromorfic function in $z$ and in fact analytic for $z\neq 0$.
If there is an infinite number of eigenvalues, then this function is zero for an infinite number of $z \in \UU(1)$ on the unit circle.Therefore, the
expression is identically to the zero function, and $T_{z,0}\smat{1\\1}$
is an eigenvector of $ T_z $ for all $z\in \CC^*$.
The eigenvalue $\lambda_z$ is then given by  
$$
\lambda_z\,=\,\pmat{1 & 0} T_{z,0}^{-1} T_z T_{z,0}\pmat{1\\1},
$$
which is also analytic in a neighborhood of the unit circle $\UU(1)$.
At a specific eigenvalue $z_0$ we have $|\lambda_{z_0}|<1$ and thus, $|\lambda_z|<1$ for an open interval on the unit circle around $z_0$.
But by Claim~1 this implies a set of non-countably many eigenvalues which can not happen on a separable Hilbert space and we have a contradiction.
This proves Claim~2 and finishes the proof of the lemma.
\end{proof}

And finally we remark that there is no essential spectrum where $|\Tr T_z|>2$.

\begin{lemma}\label{lem-5}
The set $\{z \in \UU(1)\,:\, |\Tr z|>2\, \wedge \, z \;\text{is not eigenvalue of }\Uu \}$ is in the resolvent set of $\Uu$.

\end{lemma} 
\begin{proof}
By continuity and Lemma~\ref{lem-4}, the set
$$\Ss\,=\, \{z \in \UU(1)\,:\, |\Tr z|>2\} \setminus \{z\in\UU(1)\,:\, z\; \text{eigenvalue}\}$$ 
is a finite union of open intervals on the unit circle. Moreover, for $z \in \Ss$ we find matrices $M_z$ and $\lambda_z$ with $|\lambda_z|<1$ such that
$$
M_z^{-1} T_z M_z\,=\, \pmat{ \lambda_z \\ & s_z \lambda_z^{-1}} \qtx{where} s_z=\det T_z \in \UU(1).
$$
Moreover, for $z$ varying inside compact subintervals $e^{i[a,b]}\subset \Ss$,  one may choose $M_z$ $M_z^{-1}$ so that 
$$\|M_z\|<C,\quad \|M^{-1}_z\|<C, \quad |\lambda^{-1}_z|>C>1 \qtx{uniformly for} z\in e^{i[a,b]}\,.
$$
Now let
$$
\pmat{x_z \\ y_z}\,=\,M_z^{-1} T_{z,0} \pmat{1\\1},
$$
and as $z$ is not an eigenvalue of $\Uu$, Lemma~\ref{lem-4} implies $y_z \neq 0$.Thus, $|y_z|>\varepsilon>0$ for some $\varepsilon>0$ and all
$z \in e^{i[a,b]}$. Thus, for $f\in C(\UU(1))$ supported in $e^{i[a,b]}$ we have using Theorem~\ref{th-main}
$$
\mu(f)\,=\,\lim_{n \to \infty} \int_a^b \frac{f(e^{i\varphi}){\rm d}\varphi}{\pi \|T_{z,[0,np]} \smat{1\\1}\|^2}\,=\,
\lim_{n \to \infty} \int_a^b \left\| M_z \pmat{\lambda_z^n x_z \\ \lambda_z^{-n} s_z y_z} \right\|^{-2}\,\frac{f(e^{i\varphi})\,{\rm d}\varphi}{\pi} 
$$
$$
\,\leq\, \lim_{n\to \infty} \frac{ \|f\|_\infty (b-a)}{C^{-2} \left( |\lambda^{-n}_z y_z| - |x_z \lambda_z^n| \right)^{2}}
\,\leq \, \lim_{n\to \infty} \frac{\|f\|_\infty (b-a) C^2}{(\varepsilon C^n-|x_z| C^{-n})^2}\,=\,0\,.
$$
Thus, $\supp \mu \cap \Ss=\emptyset$ and with Lemma~\ref{lem-1} the claim follows. 
\end{proof}  

\begin{proof}[Proof of Proposition~\ref{prop-sigma-ac}]
Note that by Lemma~\ref{lem-1} the spectrum of $\Uu$ has multiplicity 1 everywhere and is given by the support and measure types of $\mu$.
Then, Lemmata~\ref{lem-3}, \ref{lem-4} and \ref{lem-5} imply that there is no essential spectrum in the interior\footnote{interior with respect to the topology on the unit circle $\UU(1)$.} of
$\{z\in\UU(1)\,:\,|\Tr T_z|\geq 2\}$. So Lemma~\ref{lem-2} now shows
$$
\sigma_{ess}(\Uu)\,=\,\sigma_{ac}(\Uu)\,=\,\overline{\Sigma}\;.
$$
Together with Lemma~\ref{lem-4} we deduce that the spectrum is purely absolutely continuous within $\overline{\Sigma}$. Apart from $\overline{\Sigma}$
the spectrum contains only a finite number of eigenvalues outside $\overline{\Sigma}$. Finally, as the Hilbert space is infinite dimensional and the spectrum has multiplicity one, a finite set of eigenvalues can not be the whole spectrum. Thus, $\Sigma\neq \emptyset$.
This finishes the proof of Proposition~\ref{prop-sigma-ac}.
\end{proof}

\begin{proof}[Proof of Theorem~\ref{th-main3}]
As $\Uu_\omega$ is almost surely a compact perturbation of $\Uu$ by condition (C2), we see that as a set, the essential spectrum of
$\Uu_\omega$ is also given by $\overline{\Sigma}$, and Lemma~\ref{lem-2} now shows Theorem~\ref{th-main3}.
\end{proof}.

\appendix

\section{Facts of linear algebra}

\begin{proposition}\label{prop-A0}
For a matrix 
$$
M=\pmat{\alpha & \beta \\ \gamma & \delta}\,\in\,\CC^{2L \times 2L},
$$
with square blocks of the same size $L\times L$ where $\beta$ is invertible,  then define
$$
\varphi_\sharp(M)\,=\,\pmat{\beta^{-1} & -\beta^{-1}\alpha \\ \delta \beta^{-1} & \gamma-\delta \beta^{-1} \alpha} \qtx{and}
\varphi_\flat(M)\,=\,\pmat{\gamma-\delta\beta^{-1}\alpha & \delta \beta^{-1} \\ -\beta^{-1} \alpha & \beta^{-1}}\;.
$$
Then, 
$$
\pmat{\Psi_- \\ \Psi_+} \,=\, M \pmat{\Phi_- \\ \Phi_+}
\quad\Leftrightarrow\quad
\pmat{\Phi_+ \\ \Psi_+}\,=\,\varphi_\sharp(M) \pmat{\Psi_- \\ \Phi_-}\quad \Leftrightarrow\quad
\pmat{\Psi_+ \\ \Phi_+}\,=\,\varphi_\flat(M) \pmat{\Phi_- \\ \Psi_-}\;.
$$
Moreover, if $M\in \UU(2L)$, then $\varphi_\sharp(M), \varphi_\flat(M) \in \UU(L,L)$, where
$$
\UU(L,L)\,=\,\left\{T\in \CC^{2 \times 2}\,:\, T^*\pmat{\one \\ & -\one} T=\pmat{\one \\ & -\one} \right\}.
$$
The inverse maps are given by
$$
\varphi_\sharp^{-1} \pmat{A&B\\C&D}\,=\,\pmat{-BA^{-1} & A^{-1} \\ D-CA^{-1}B & CA^{-1}}\;, \qquad
\varphi_\flat^{-1} \pmat{A&B\\ C&D}\,=\,\pmat{-CD^{-1} & D^{-1} \\ A-BD^{-1}C & BD^{-1}}.
$$
\end{proposition}
\begin{proof}
Resolving the linear system 
$$\Psi_-=\alpha \Phi_-+\beta \Phi_+,\qquad \Psi_+=\gamma \Phi_- +\delta \Phi_+,$$  for $\Psi_+,\; \Phi_+$ gives the equivalent system
$$
\Phi_+ = \beta^{-1} \Psi_- - \beta^{-1} \alpha \Phi_-\;,\qquad 
\Psi_+=\delta \beta^{-1} \Psi_- +(\gamma-\delta \beta^{-1}\alpha) \Phi_-\;.
$$
For the second part note $\varphi_\flat(M)=\smat{0 & I \\ I & 0} \varphi_\sharp(M) \smat{0 & I \\ I & 0} \in \UU(L,L)$ if and only if $\varphi_\sharp(M) \in \UU(L,L)$ and using $B=\beta^{-1}$, $\one=\alpha^*\alpha+\gamma^* \gamma=\delta^*\delta+\beta^*\beta$ and $\alpha^*\beta+\gamma^*\delta=0$ 
for $M\in\UU(2L)$, we find
$$
\varphi_\sharp(M)^* \pmat{\one \\ & -\one} \varphi_\sharp(M)\,=\,\pmat{B^* &  B^* \delta^*  \\ - \alpha^* B^* & \gamma^*-\alpha^* B^* \delta^*}\pmat{B & -B \alpha \\ - \delta B & \delta B \alpha - \gamma}\,=\,\pmat{X & Y \\ Y^* & Z},
$$
where 
\begin{align*}
X
&=B^*(\one-\delta^*\delta) B=B^* \beta^* \beta B = \one\\
Y
&=B^*(\delta^* \delta-\one) B\alpha-B^* \delta^* \gamma =-B^* (\beta^*\alpha+\delta^* \gamma)=0\\
Z
&=-\gamma^* \gamma+\alpha^* B^*(\one-\delta^* \delta)B\alpha+\gamma^* \delta B \alpha+\alpha^*B^* \delta^* \gamma\\
&=-\gamma^*\gamma+\alpha^*(\one-\beta B - B^* \beta^*)\alpha=-(\gamma^* \gamma+\alpha^*\alpha)=-\one\;.
\end{align*}
This shows that  $\varphi_\sharp(M)\in\UU(L,L)$ for $M\in\UU(2L)$.
\end{proof}

\begin{proposition}\label{prop-A0.5}
Let 
$$
M=\pmat{\alpha & \beta \\ \gamma & \delta}\,\in\,\CC^{2L \times 2L},
$$
with square blocks of the same size $L\times L$ as above and assume $\beta$ is invertible.

\begin{enumerate}[{\rm a)}]
\item Assume moreover that $U, V$ are invertible $L\times L$ matrices. Then
$$
\varphi_\sharp\left(\pmat{U & \nul \\ \nul & V} M \right)\,=\,\pmat{\one & \nul \\ \nul & V} \varphi_\sharp(M) \pmat{U^{-1} & \nul \\ \nul & \one},
$$
and
$$
\varphi_\sharp\left(M \pmat{U & \nul \\ \nul & V} \right)\,=\,\pmat{V^{-1} & \nul \\ \nul & \one} \varphi_\sharp(M) \pmat{\one & \nul \\ \nul & U}.
$$
\vspace{.2cm}
\item Assume moreover that  $\one+M$ is invertible, then
$$
\varphi_\sharp\left(\one+M)^{-1} M \right)\,=\,\pmat{\one & \one  \\ \nul &  \one} \varphi_\sharp(M) \pmat{\one & \nul \\ -\one & \one}.
$$
\end{enumerate}
\end{proposition}

\begin{proof}
Part (a) is easy to check with the definition of $\varphi_\sharp$  in Proposition~\ref{prop-A0}.\\
For part (b) let 
$$
(\one+M)^{-1} M \pmat{\Phi_- \\ \Phi_+}\,=\,\pmat{\Psi_-\\ \Psi_+} ,  $$
then 
$$
M \pmat{\Phi_--\Psi_- \\  \Phi_+-\Psi_+}\,=\,\pmat{ \Psi_- \\  \Psi_+},
$$
and thus
$$
\pmat{ \Phi_+- \Psi_+ \\  \Psi_+}\,=\,\varphi_\sharp(M) \pmat{ \Psi_- \\  \Phi_--  \Psi_-}\,=\,\varphi_\sharp(M) \pmat{ \one & \nul \\ -\one & \one} \pmat{\Psi_- \\ \Phi_-},
$$
which gives
$$
\pmat{\Phi_+ \\ \Psi_+}\,=\,\pmat{\one &  \one \\ \nul & \one} \pmat{ \Phi_+- \Psi_+ \\  \Psi_+}\,=\,
 \pmat{\one &  \one \\ \nul & \one} \varphi_\sharp(M)  \pmat{\one & \nul \\ -\one & \one} \pmat{\Psi_- \\ \Phi_-}.
$$
With Proposition~\ref{prop-A0} this finishes the proof.
\end{proof}

\begin{proposition}\label{prop-A1}
Let 
$$
U=\pmat{A&B\\C&D}\,\in\, \UU(n),
$$
be a unitary $n \times n$ matrix split in blocks of sizes $n_1, n_2$, meaning $A\in \CC^{n_1\times n_1}$, $B \in \CC^{n_1 \times n_2}$, $C \in \CC^{n_2 \times n_1}$, $D \in \CC^{n_2 \times n_2}$.
Let $P\in\UU(n_2)$.  Then, the following holds.
\begin{enumerate}[{\rm a)}]
\item $D-P$ is invertible $\quad\Leftrightarrow \quad$ $\pmat{A&B\\C&D-P}$ is invertible.\\[.2cm]
Moreover, if $B$ has trivial kernel then $D-P$ is invertible (for any $P\in\UU(n_2)$).\\[.2cm]
\item Assume that $P \in \UU(n_2)$ and that
$D-P$ is invertible.
Then,  the Schur complement
$$
A-B(D-P)^{-1}C\,\in\,\UU(n_1),
$$
is unitary.  As a consequence,  
$$
\pmat{\one & \nul} \pmat{A&B\\C&D-P}^{-1} \pmat{\one \\ \nul} \,\in\,\UU(n_1),
$$
is unitary.
\item More general, let $P,\,W \in \CC^{n_2 \times n_2}$ such that $W^*P=\one_{n_2}$, and $D-P$ and $D-W$ are invertible, then
$$
(A-B(D-W)^{-1}C)^* (A-B(D-P)^{-1}C)\,=\,\one_{n_1}.
$$
\end{enumerate}
\end{proposition}
\begin{proof}
First note that if  $(D-P)\varphi=\nul$, then
$$
\|\varphi\|^2=\left\|U\pmat{\nul \\ \varphi}\right\|^2=\|B\varphi\|^2 + \|D\varphi\|^2=\|B\varphi_2\|+\|P\varphi\|^2=\|B\varphi\|^2+\|\varphi\|^2,
$$
meaning $B\varphi=\nul$. 
Therefore, 
$$
\pmat{A & B \\ C & D-P} \pmat{\nul \\ \varphi}\,=\,\nul.
$$
In particular, if $B$ has trivial kernel, $\varphi=\nul$ and, thus, $D-P$ is invertible.\\
On the other hand,  let $\pmat{A & B \\ C & D-P}\pmat{\varphi_1 \\ \varphi_2}=\nul$, then
$$
\|\varphi_1\|^2+\|\varphi_2\|^2=\left\| \pmat{\varphi_1 \\ \varphi_2}\right\|^2=\left\| U\pmat{\varphi_1 \\ \varphi_2}\right\|^2=\left\| \pmat{\nul \\ P\varphi_2}\right\|^2 =\|\varphi_2\|^2.
$$
Therefore, $\varphi_1=\nul$ and $D\varphi_2=P\varphi_2$.
Hence,  
$$
\ker\pmat{A& B \\ C & D-P}\,=\,\left\{ \pmat{\nul\\ \varphi}\,:\,\varphi \in \ker(D-P) \;\right\}.
$$
For part b) note that the first part of b) follows from c). 
Using the Schur complement identity
$$
\pmat{\one & \nul} \pmat{A & B \\ C & D-P}^{-1} \pmat{\one \\ \nul}\,=\,(A-B(D-P)^{-1}C)^{-1},
$$
the second statement follows.\\[.2cm]
For part c) by unitarity of $U$ we get
$$
A^*A+C^*C=\one_{n_1}\;,\quad B^*B+D^*D=\one_{n_2}\;,\quad A^*B=-C^*D,
$$
and using $W^*P=\one$ we obtain 
\begin{align*}
 &(A-B(D-W)^{-1}C)^*(A-B(D-P)^{-1}C)\\
& =A^*A-A^*B(D-P)^{-1}C -C^*(D^*-W^*)^{-1}B^*A+
C^*(D^*-W^*)^{-1}B^*B(D-P)^{-1}C\\
&= A^*A+C^*D (D-P)^{-1}C+C^*(D^*-W^*)^{-1}D^*C+
C^*(D^*-W^*)^{-1}(\one-D^*D)(D-P)^{-1} C\\
&= A^*A+C^*(D^*-W^*)^{-1} \big( (D^*-W^*)D+D^*(D-P)+\one-D^*D\big)(D-P)^{-1}C\\
&=A^*A + C^*(D^*-W^*)^{-1}\big((D^*-W^*)(D-P)\big)(D-P)^{-1}C=A^*A+C^*C=\one .
\end{align*}

\end{proof}

\begin{proposition}\label{prop-U11}
For $T\in\UU(1,1)$ we have $|\det T|=1$ and $\tfrac{(\Tr T)^2}{ \det T} \geq 0$.
In particular, there is $\chi \in \RR$ and $\lambda\in \CC$, such that the eigenvalues of $T$ are given by
$e^{i\chi}\lambda$ and $e^{i\chi} \lambda^{-1}$ with $\lambda+\lambda^{-1}\in\RR$.
In particular, $\lambda\in\RR$ if $|\Tr T|=|\lambda+\lambda^{-1}|\geq 2$, and $\lambda\in\UU(1)$ if $|\Tr T|\leq 2$.
\end{proposition}
\begin{proof}
If
$$
T=\pmat{a&b\\c&d}, \quad G=\pmat{1 & 0 \\ 0 & -1},
$$
then  $T^*GT=G$ which means $-1=\det(G)=\det(G)|\det T|^2$ implying $|\det(T)|=1$. 
Moreover, we get $|a|^2-|c|^2=1=|d|^2-|b|^2$ and $\bar a b = \bar c d$.
On the other hand, $T^*GTG=G^2=I$ and therefore $TGT^*G=I$ giving $TGT^*=G$, which gives
$|a|^2-|b|^2=1=|d|^2-|c|^2$. Hence, $|a|=|d|=\cosh(\gamma)$ and $|b|=|c|=\sinh(\gamma)$ for some $\gamma \geq 0$.\\
Then,
$$
a=\cosh(\gamma) e^{i\varphi_1}, \; b=\sinh(\gamma) e^{i\varphi_2},\; c= \sinh(\gamma) e^{i\varphi_3},\; d= \cosh(\gamma)e^{i\varphi_4},
$$
where now $\bar a b= \bar c d$ leads to
$$
e^{-i\varphi_1+i\varphi_2}\,=\,e^{-i \varphi_3+i\varphi_4} \qtx{implying} e^{i\varphi_1+i\varphi_4}=e^{i\varphi_2+i\varphi_3}\,.
$$
Using this relation, we get $\det(T)= e^{i\varphi_1+i\varphi_4}$ and
$$
\frac{(\Tr T)^2}{ \det T}  = \frac{\cosh^2(\gamma) (e^{i\varphi_1}+e^{i\varphi_4})^2}{e^{i\varphi_1+i\varphi_4}}\,=\,
\cosh^2(\gamma) \left( 2+ 2 \cos(\varphi_1-\varphi_4)\right)\,\geq\, 0\;.
$$
Now with $2\chi=\varphi_1+\varphi_4$ we have $\det(T)=e^{2i\chi}$ and the eigenvalues are of the form $e^{i\chi} \lambda$ and $e^{i\chi} \lambda^{-1}$
for some complex $\lambda$, and the trace is equal to $e^{i\chi}(\lambda+\lambda^{-1})$, which leads to $(\lambda+\lambda^{-1})^2=\frac{(\Tr T)^2}{\det(T)}\geq 0$.
Therefore, $\lambda+\lambda^{-1} \in \RR$.
\end{proof}

\section*{Declarations}

 On behalf of all authors, the corresponding author states that there is no conflict of interest.
There is no associated data for this research.

This work has been supported by the Chilean grants  FONDECYT 1230949, 
1211576, 1211189, 1201836.

\end{document}